\documentclass[a4paper,fleqn,11pt,preprint]{cas-sc}

\usepackage[square,numbers]{natbib}

\usepackage{amsmath}
\usepackage{empheq}
\usepackage{mathtools}
\usepackage{amsthm}
\newtheorem{theorem}{Theorem}[section]
\newtheorem{lemma}[theorem]{Lemma}
\usepackage{rotating}
\usepackage{ifthen}
\usepackage{extarrows}
\usepackage{multirow}  
\usepackage{array}     
\usepackage{graphicx}
\usepackage{lscape}
\usepackage{xcolor}
\usepackage{setspace}
\usepackage{url}

\newcommand{\order}[1]{\Theta(#1)}

\newcommand{\pol}[2]{
\ifthenelse{\equal{#1}{}}{N^{#2}}{#1 \cdot N^{#2}}
}

\def\tsc#1{\csdef{#1}{\textsc{\lowercase{#1}}\xspace}}
\tsc{WGM}
\tsc{QE}
\tsc{EP}
\tsc{PMS}
\tsc{BEC}
\tsc{DE}

\begin{document}
\doublespacing

\let\WriteBookmarks\relax
\def\floatpagepagefraction{1}
\def\textpagefraction{.001}
\shorttitle{Speedup and efficiency of computational parallelization}
\shortauthors{Schryen}



\title [mode = title]{Speedup and efficiency of computational parallelization: A unifying approach and asymptotic analysis}

\author[1]{Guido Schryen}[]
\cormark[1]
\ead{guido.schryen@uni-paderborn.de}

\address[1]{Department of Management Information Systems, Paderborn University, Warburger Strasse 100, Paderborn 33098, Germany}

\cortext[cor1]{Corresponding author}

\begin{abstract}
In high performance computing environments, we observe an ongoing increase in the available number of cores. For example, the current TOP500 list reveals that nine clusters have more than 1 million cores. This development calls for re-emphasizing performance (scalability) analysis and speedup laws as suggested in the literature (e.g., Amdahl’s law  and Gustafson’s law), with a focus on asymptotic performance. Understanding speedup and efficiency issues of algorithmic parallelism is useful for several purposes, including the optimization of system operations, temporal predictions on the execution of a program, the analysis of asymptotic properties, and the determination of speedup bounds. However, the literature is fragmented and shows a large diversity and heterogeneity of speedup models and laws. These phenomena make it challenging to obtain an overview of the models and their relationships, to identify the determinants of performance in a given algorithmic and computational context, and, finally, to determine the applicability of performance models and laws to a particular parallel computing setting.
In this work, I provide a generic speedup (and thus also efficiency) model for homogeneous computing environments. My approach generalizes many prominent models suggested in the literature and allows showing that they can be considered special cases of a unifying approach. The genericity of the unifying speedup model is achieved through parameterization. Considering combinations of parameter ranges, I identify six different asymptotic speedup cases and eight different asymptotic efficiency cases. Jointly applying these speedup and efficiency cases, I derive eleven scalability cases, from which I build a scalability typology. Researchers can draw upon my suggested typology to classify their speedup model and to determine the asymptotic behavior when the number of parallel processing units increases. \color{red} Also, the description of two computational experiments demonstrates the practical application of the model and the typology. \color{black} In addition, my results may be used and extended in future research to address various extensions of my setting.
 
\end{abstract}



\maketitle
\doublespacing



\section{Introduction}
\label{sec:introduction}

 Parallel computing has become increasingly important for solving hard computational problems in a variety of scientific disciplines and industrial fields. The large diversity and deployment of parallel computing across disciplines, including  artificial intelligence, arts and humanities, computer science, digital agriculture, earth and environmental sciences, economics, engineering, health sciences, mathematics, and natural sciences, is mirrored in usage statistics published by supercomputer clusters (e.g., \citep{NCSA21,JSC22}). This ongoing progress in computational sciences through parallelization has been fostered through the end of exponential growth in single processor core performance \citep{fuller2011future} and the availability of high performance computing (HPC) infrastructures, tools, libraries, and services as commodity goods offered by computing centers of universities, public cloud providers, and open
source communities.
\color{red}
The development of parallel computing has been accompanied by the study of its performance. Generally speaking, performance in parallel computing refers to the behavior of a parallel computing system in processing specified tasks with respect to the amount of resources, such as parallel computing units, that are used or available. It encompasses a variety of metrics and concepts, including speedup, efficiency, load balancing, and communication overhead, among others. Performance has also been studied as the amount of parallel computing resources grows to infinity (asymptotic performance), resulting in a variety of speedup laws. For an introduction to performance in parallel computing, see, for example, \citep[ch. 5]{Grama2003}. 
\color{black}

Beyond these developments, the number of cores available as parallel processing units has increased substantially over the past years. While the statistics of the TOP500 list (as of June 2022) shows values of 35,339.2 (10th percentile), 67,328 (median), and 225,465.6 (90th percentile), the corresponding values of the lists as of June 2017 and June 2012 amount to (16,545.6; 36,000; 119,808) and (6,776; 13,104; 37,036.8), respectively \citep{TOP500}. In addition, in contrast to the lists of 2012 and 2017, which both include only one site with more than 1 million cores, the current list shows that nine clusters have more than 1 million cores. This enormous growth in the number of cores which are available for parallel processing calls for re-emphasizing asymptotic performance analysis (e.g., \citep{Che2014,al2020amdahl}) and speedup laws as suggested in the literature (e.g., Amdahl's law \citep{amdahl1967validity} and Gustafson's law \citep{gustafson1988reevaluating}).



In general, studying performance of algorithmic parallelism is useful for several purposes; these include optimizing system operations via design-time and run-time management (e.g., \citep{hill2008amdahl,zidenberg2012multiamdahl,xia2017voltage,zhuravlev2012survey,chen2018scheduling}), making temporal predictions on the execution of a program (e.g., \citep{rai2010performance,ababei2018survey}), and analyzing asymptotic speedup and efficiency properties as well as determining speedup and efficiency bounds (e.g., \citep{sun2010reevaluating,Che2014}). In this article, I focus on the two latter purposes, which have been addressed only rarely in the literature.   

Analyzing performance of parallel algorithms is challenging as it needs to account for diversity in several regards. For example, existing speedup models and laws make different assumptions with respect to the homogeneity/heterogeneity of parallel processing units, variations of workloads, and methodological characteristics and application fields of algorithms (e.g., optimization, data analytics, simulation). This heterogeneity has resulted in a landscape of many speedup models and laws, which, in turn, makes it difficult to obtain an overview of the models and their relationships, to identify the determinants of performance in a given algorithmic and computational context, and, finally, to determine the applicability of performance models and laws to a particular parallel computing setting.

My focus lies on the development of a generic and unifying speedup and efficiency model for homogeneous parallel computing environments. I consider a range of determinants of speedup covered in the literature and prove that existing speedup laws can be derived from special cases of my model. My model depends neither on specific system architectures, such as symmetric multiprocessing (SMP) systems or graphics processing units (GPU), nor on software properties, such as critical regions; I rather perform a theoretical analysis \color{red} although I also conduct computational experiments to demonstrate the application of the model. \color{black} I further focus on the analysis of asymptotic properties of the suggested model to study speedup and efficiency limits and bounds in the light of a computing future with an increasing number of parallel processing units.

My results contribute to research on the performance (in terms of scalability) of computational parallelization in homogeneous computing environments in several regards: (1) I suggest a generic speedup and efficiency model which accounts for a variety of conditions under which parallelization occurs so that it is broadly applicable. This wide scope allows conducting performance analysis in many of those cases which are not covered by existing models and laws with restrictive assumptions. (2) I generalize the fragmented landscape of speedup and efficiency models and results, and I provide a unifying speedup and efficiency model which allows overcoming the perspective of conflicting speedup models by showing that these models can be interpreted as special cases of a more universal model. (3) From my asymptotic analysis, I derive a typology of scalability (speedup and efficiency), which researchers may use to classify their speedup model and/or to determine the asymptotic behavior of their particular application. I also provide a theoretical basis for explaining sublinear, linear and superlinear speedup and efficiency and for deriving speedup and efficiency bounds in the presence of an enormous growth of the number of available parallel processing units. \color{red}  (4) I demonstrate the practical application of the speedup and efficiency model and the typology with computational experiments on matrix multiplication and lower-upper matrix decomposition. \color{black} To sum up, I consolidate prior research on performance in homogeneous parallel computing environments and I provide a theoretical understanding of quantitative effects of various determinants of asymptotic performance in parallel computing.

The remainder of the article is structured as follows: In Section 2, I provide a brief overview of the foundations of speedup and efficiency analysis in parallel computing. I proceed in Section 3 with the suggestion of a generic speedup and efficiency model. In Section 4, I perform a mathematical analysis of my model in order to determine theoretical speedup and efficiency limits. \color{red} I describe the computational experiments in Section 5. In Section 6, I discuss the application of the proposed model and scalability typology, and I consider parallelization overhead. Finally, I provide conclusions of my research in Section 7. \color{black}
\section{Foundations of speedup and efficiency analysis}
\label{sec:foundationsSpeedupEfficiencyAnalysis}

The main purpose of parallel computation is to take advantage of increased processing power to solve problems faster or to achieve better solutions. The former goal is referred to as \textit{scalability}, and scalability measures fall into two main groups: \textit{speedup} and \textit{efficiency}. Speedup $S(N)$ is defined as the ratio of sequential computation time $T(1)$ to parallel computation time $T(N)$ needed to process a task with given workload when the parallel algorithm is executed on $N$ parallel processing units (PUs) (e.g., cores in a multicore processor architecture); i.e.,
\begin{equation}
S(N):=\frac{T(1)}{T(N)},\; T: \mathbb{N} \rightarrow R^{> 0}\label{eq:definitionSpeedupBasic}
\end{equation}The sequential computation time $T(1)$ can be measured differently, leading to different interpretations of speedup \citep{barr1993reporting}: When $T(1)$ refers to the fastest serial time achieved on any serial computer, speedup is denoted as \textit{absolute}. Alternatively, it may also refer to the time required to solve a problem with the parallel program on one of the parallel PUs. This type of speedup is qualified as \textit{relative}. In this work, I focus on relative speedup. 

As speedup relates the time required to process a given workload on a single PU to the time required to process the same workload on $N$ PUs, you need to determine this workload. It is usually divided into two sub-workloads, the sequential workload and the parallelizable workload. While the former is inherently sequential and necessarily needs to be executed on a single PU, the latter can be executed in parallel on several PUs. Independent of the number of available parallel PUs $N$, the time required to solve a task is the sum of the time to handle the sequential workload and the time to handle the parallelizable workload  of the given task. When only a single PU is available, the time for the sequential workload $s$ and for the parallelizable workload $p$ are usually normalized by setting $s+p=1$; i.e., $s$ and $p$ represent the sequential and the parallelizable fractions of the overall execution time.

For some applications, it is useful to consider a fixed workload (e.g., when solving an instance of an optimization problem), which is independent of the number of parallel PUs ($N$) available, and then to analyze how computation of the fixed workload on a single PU can be speeded up by using multiple PUs. Speedup models of this type are referred to as \textit{fixed-size models}, such as Amdahl's law \citep{amdahl1967validity}. For other applications (e.g., when analyzing data), is more appropriate to use the availability of $N$  PUs to solve tasks with workloads which increase depending on $N$. Then, scalability analysis deals with investigating how computation of the variable workload on one PU can be speeded up by using multiple PUs.  Speedup models of that type are referred to as \textit{scaled-size models}, such as Gustafson's law \citep{gustafson1988reevaluating}.  

With varying number of PUs $N$, both the sequential and parallelizable workload may be considered scalable. It is common in the literature to introduce two workload scaling functions $f(\cdot), g(\cdot)$ with $f,g:\mathbb{N} \rightarrow R^{> 0}$ for the sequential and parallelizable workload, respectively; i.e.; the (normalized) time to process the sequential and the parallelizable workload on a single PU are $s\cdot f(N)$ and $p\cdot g(N)$, respectively. Thus, the (normalized) time to process the overall workload on a single PU amounts to $s \cdot f(N) + p \cdot g(N)$. Usually, it is assumed that $f(1)=g(1)=1$ so that $T(1)=s+p=1$ holds; however, my workload scaling functions do not require to meet this assumption.\footnote{The option to have values $f(1)\neq 1$ and/or $g(1)\neq 1$ allows scaling both fractions $s$ and $p$, which may be useful when an overall workload to be executed on a machine $A$ (with $p+s=1$) is now executed on a different machine $B$ on which the times to execute the serial and the parallelizable workload are scaled at either the same or different rates.} An example of using a scaling function for the sequential workload can be found in the scaled speedup model suggested by \citet[p. 31f\/f]{schmidt2017parallel}. While scaling functions for sequential workloads can be found only rarely, scaling functions for parallelizable workloads are much more common; see, for example, the  scale-sized speedup model of Gustafson  \citep{gustafson1988reevaluating}, the memory-bound speedup model of Sun and Ni \citep{sun1990another,sun1993scalable,sun2010reevaluating}, the generalized scaled speedup model of \citet{juurlink2012amdahl} and the scaled speedup model of \citet[p. 31f\/f]{schmidt2017parallel}. A discussion of the relationship between problem size and the number of PUs $N$ can be found in \citep[p. 32f]{trobec2018introduction}.


While the time required to process the sequential workload is independent of the number of PUs $N$, the time required to process the parallelizable workload depends on $N$ as this workload can be processed in parallel. Usually, the parallelizable workload is considered to be equally distributed on $N$ PUs, resulting in the (normalized) time $\frac{p\cdot g(N)}{N}$ to handle the parallelizable workload. However, there are tasks possible when the time required to handle the parallelizable workload is affected due to its actual parallel execution; for example, when a mathematical optimization problem, such as a mixed-inter linear program (MILP), is solved with a parallelized branch-and-bound algorithm, then good bounds may be found early so that the branch-and-bound tree does not grow as large as with the sequential execution of the branch-and-bound algorithm. This effect may result in a denominator function which is not identical to $N$ and allows explaining superlinear speedup as it has been observed in the literature (e.g., \citep{Rauchecker_Schryen_2019,dabah2022efficient,Gonggiatgul2023}). I account for this effect with a scaling function $h(\cdot)$, with $h :\mathbb{N} \rightarrow R^{> 0}$.   

Finally, processing one single large task on several parallel PUs involves some sort of overhead, which is often rooted in initialization, communication, and synchronization efforts \citep{yavits2014effect,huang2013extending,flatt1989performance}. I account for the additional time required for these efforts with an overhead function $z(\cdot)$,  $z :\mathbb{N} \rightarrow R^{> 0}$.  

The abovementioned workloads and temporal effects are visualized in Figure \ref{fig:speedupVisualization}. The resulting general speedup equation is the given by 
\begin{equation}
\label{eq:generalSpeedupEquation}
S(N)= \frac{T(1)}{T(N)}=\frac{s\cdot f(1)+p\cdot g(1)}{s\cdot f(N)+\frac{p\cdot g(N)}{h(N)}+z(N)}\; , N\in \mathbb{N}
\end{equation}  


\begin{figure}
\fbox{\includegraphics[width=\linewidth]{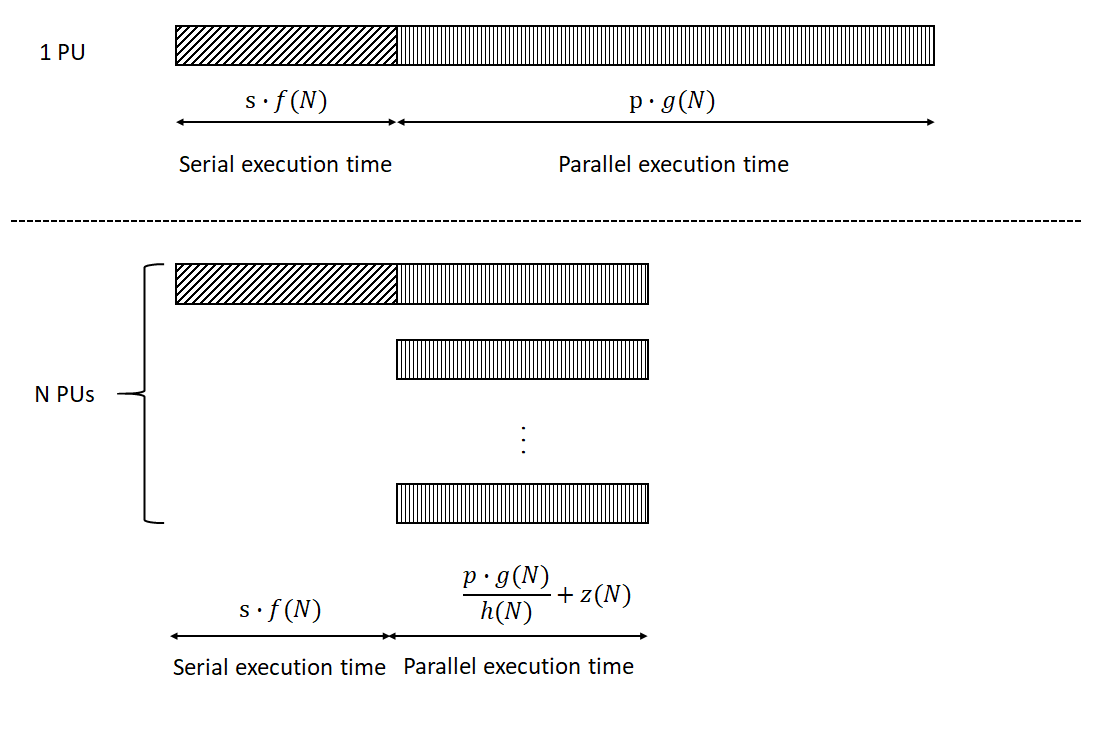}}
\caption{Workloads and temporal effects of parallelization}
\label{fig:speedupVisualization}
\end{figure}

Note that the speedup equation given in (\ref{eq:generalSpeedupEquation}) is a generalization of several well known speedup models, including those used in Amdahl's law \citep{amdahl1967validity} (set $f(N)=g(N)=1, h(N)=N, z(N)=0$), Gustafson's law \citep{gustafson1988reevaluating} (set $f(N)=1, g(N)=h(N)=N, z(N)=0$), and the generalized scaled speedup model \citep{juurlink2012amdahl} (set $f(N)=1, g(N)=N^{0.5}, h(N)=N, z(N)=0$). 

Based upon speedup $S(N)$, efficiency $E(N)$ relates speedup to the number of parallel PUs used to achieve this speedup, and it is defined by
\begin{equation}
\label{eq:generalEfficiencyEquation}
E(N)= \frac{S(N)}{N}=\frac{T(1)}{N \cdot T(N)}=\frac{s\cdot f(1)+p\cdot g(1)}{N \cdot (s\cdot f(N)+\frac{p\cdot g(N)}{h(N)}+z(N))}\; , N\in \mathbb{N}
\end{equation}

\section{A generic speedup and efficiency model}
\label{sec:genricSpeedupEfficiencyModel}
Based upon the general speedup equation (\ref{eq:generalSpeedupEquation}), I derive a generic speedup and efficiency model, which I use in the remainder of this article to analyze its asymptotic behavior. The generic speedup model uses power functions for $f(\cdot),g(\cdot),h(\cdot)$ and ignores any overhead induced through parallelization. The use of power functions is widely adopted in the literature, included in many prominent speedup models \citep{amdahl1967validity,gustafson1988reevaluating,schmidt2017parallel,sun1993scalable,juurlink2012amdahl} and is based on the assumption that many algorithms have a polynomial complexity in terms of computation and memory requirement \citep[p. 184]{sun2010reevaluating}. As I focus on the analysis of the asymptotic behavior, I always take the highest degree term. The motivation for neglecting any parallelization overhead (i.e., $z(N)=0\; \forall N \in \mathbb{N}$), as it is done in many, if not most speedup and efficiency models in the literature, is manifold: First, the overhead is often unknown. Second, omitting an overhead term simplifies computations and provides a basis for developing laws which include an overhead function $z(N) \neq 0$. Third, speedup and efficiency values determined without considering overhead represent upper bounds for practically achievable speedup and efficiency values when overhead occurs.

I use the following power functions:
\begin{equation}
\label{def:powerFunctions}
f(N):=c_f N^{\alpha_f}\qquad g(N):=c_g N^{\alpha_g}\qquad h(N):=c_h N^{\alpha_h},\qquad c_f,c_g,c_h > 0,\; \alpha_f,\alpha_g,\alpha_h \geq 0
\end{equation}

and yield the following generic speedup equation (for $N>1$)\footnote{The applicability of this speedup equation to $N=1$ would require setting $c_h=1$. In order to allow using an arbitrary coefficient in the power function $h(\cdot)$, I require $N$ to be larger than $1$.}
\begin{equation}
\label{eq:genericSpeedupEquation}
S(N)=\frac{s\cdot f(N)+p\cdot g(N)}{s\cdot f(N)+\frac{p\cdot g(N)}{h(N)}+z(N)}=\frac{\pol{s \cdot c_f}{\alpha_f} + \pol{p \cdot c_g}{\alpha_g}}{\pol{s \cdot c_f}{\alpha_f} + \frac{\pol{p \cdot c_g}{\alpha_g}}{\pol{c_h}{\alpha_h}}}=\frac{\pol{s \cdot c_f}{\alpha_f} + \pol{p \cdot c_g}{\alpha_g}}{\pol{s \cdot c_f}{\alpha_f} + \pol{\frac{p \cdot c_g}{c_h}}{(\alpha_g - \alpha_h)}}
\end{equation}

and the following efficiency equation (for $N>1$):
\begin{equation}
\label{eq:genericEfficiencyEquation}
E(N)=\frac{S(N)}{N}=\frac{\pol{s \cdot c_f}{\alpha_f} + \pol{p \cdot c_g}{\alpha_g}}{\pol{s \cdot c_f}{(\alpha_f + 1)} + \pol{\frac{p \cdot c_g}{c_h}}{(\alpha_g - \alpha_h +1)}}
\end{equation}

The generic speedup equation given in (\ref{eq:genericSpeedupEquation}) generalizes several well-known speedup equations and laws suggested in the literature (see Table \ref{tab:generalizationSpeedupModelsLiterature}).

\begin{table}
\center
\caption{Instantiations of generic speedup model}
\label{tab:generalizationSpeedupModelsLiterature}\begin{spacing}{1.2}
\begin{tabular}{|l|c|p{6cm}|}
\hline
Parameter values & Speedup equation & Speedup model or law\\
\hline
 $c_f=c_g=c_h=\alpha_h=1, \alpha_f=\alpha_g=0$ & $S(N)=\frac{s + p}{s + \frac{p}{N}}$ & \textit{Amdahl's law} \citep{amdahl1967validity}\\[5mm]
 $c_f=c_g=c_h=\alpha_g=\alpha_h=1, \alpha_f=0$ & $S(N)=\frac{s + p \cdot N}{s + p}$ & \textit{Gustafson's law} \citep{gustafson1988reevaluating}\\[5mm]
$c_h=\alpha_h=1$ & $S(N)=\frac{s \cdot \pol{c_f}{\alpha_f} + p \cdot \pol{c_g}{\alpha_g}}{s \cdot \pol{c_f}{\alpha_f} + p \cdot \pol{c_g}{(\alpha_g-1)}}$ & \textit{Scaled speedup model} \citep{schmidt2017parallel}\\[5mm] 
& & (under the assumption that the sequential and parallel workloads are given by power functions $\pol{c_f}{\alpha_f}$ and $\pol{c_g}{\alpha_g}$, resp.)\\[5mm]
$c_f=c_g=c_h=\alpha_h=1, \alpha_f = 0$ & $S(N)=\frac{s + p \cdot \pol{}{\alpha_g}}{s + p \cdot \pol{}{(\alpha_g-1)}}$ & \textit{Sun and Ni's law} \citep{sun2010reevaluating,sun1990another,sun1993scalable}\\
& & (under the assumption that the parallel workload is given by a power function $\pol{}{\alpha_g}$)\\[5mm]  
$c_f=c_g=c_h=\alpha_h=1, \alpha_f = 0, \alpha_g=\frac{1}{2}$ & $S(N)=\frac{s + p \cdot \pol{}{\frac{1}{2}}}{s + \frac{p}{\pol{}{\frac{1}{2}}}}$ & \textit{Generalized scaled speedup model} \citep{juurlink2012amdahl}\\[5mm]
 \hline
\end{tabular}
\end{spacing}
\end{table}

\section{Theoretical speedup and efficiency limits}
\label{sec:theoreticalSpeedupEfficiencyLimits}

\subsection{Asymptotic speedup}
\label{sec:asymptoticSpeedup}

As I am interested in asymptotic speedup, I determine limits for $N\rightarrow \infty$. I rewrite the generic speedup equation (\ref{eq:genericSpeedupEquation}) as follows:
\begin{equation}
\label{eq:genericSpeedupEquationRewritten}
 S(N)=\underbrace{\frac{\pol{s \cdot c_f}{\alpha_f} }{\pol{s \cdot c_f}{\alpha_f} + \pol{\frac{p \cdot c_g}{c_h}}{(\alpha_g - \alpha_h)}}}_{\substack{(I)}} + \underbrace{\frac{\pol{p \cdot c_g}{\alpha_g}}{\pol{s \cdot c_f}{\alpha_f} + \pol{\frac{p \cdot c_g}{c_h}}{(\alpha_g - \alpha_h)}}}_{\substack{(II)}}
\end{equation}

For term (I), I yield the following limits (the proof can be obtained from equations (\ref{eq:I_Def})-(\ref{eq:IDef_c_Final}) in Appendix \ref{appendixSec:A}):
\begin{empheq}[left={\lim_{N\rightarrow \infty} (I)=}\empheqlbrace]{align}
1&,\; \alpha_f > \alpha_g - \alpha_h \label{eq:I1}\\
\frac{s \cdot c_f}{s \cdot c_f + \frac{p \cdot c_g}{c_h}} & ,\; \alpha_f = \alpha_g - \alpha_h \label{eq:I2}\\
0                 & ,\; \alpha_f < \alpha_g - \alpha_h \label{eq:I3}
\end{empheq}

For term (II), I yield the following limits (the proof can be obtained from equations (\ref{eq:II_Def_a}) - (\ref{eq:II_Def_b3c_Final}) in Appendix \ref{appendixSec:A}):
\begin{empheq}[left={\lim_{N\rightarrow \infty} (II)=}\empheqlbrace]{align}
\frac{p \cdot c_g}{s \cdot c_f}&,\; \alpha_g = \alpha_f, \alpha_h > 0 \label{eq:II1}\\
\frac{p \cdot c_g}{s \cdot c_f + \frac{p \cdot c_g}{c_h}} & ,\; \alpha_g = \alpha_f, \alpha_h = 0 \label{eq:II2}\\
0 & ,\; \alpha_g < \alpha_f \label{eq:II3}\\
\infty \;(\Theta (N^{\alpha_h})) & ,\; \alpha_g > \alpha_f, \alpha_h > 0, \alpha_g - \alpha_h - \alpha_f \geq 0 \label{eq:II4}\\
\infty \;(\Theta (N^{(\alpha_g - \alpha_f)})) & ,\; \alpha_g > \alpha_f, \alpha_h > 0, \alpha_g - \alpha_h - \alpha_f < 0 \label{eq:II5}\\
c_h & ,\; \alpha_g > \alpha_f, \alpha_h = 0 \label{eq:II6}
\end{empheq}

Aggregating the above given limits for terms (I) and (II), yields the following limits for the speedup given in equations (\ref{eq:genericSpeedupEquation}) and (\ref{eq:genericSpeedupEquationRewritten}):

\begin{empheq}[left={\lim_{N \rightarrow \infty} S(N)=}\empheqlbrace]{align}
\frac{p \cdot c_g}{s \cdot c_f} + 1 = \frac{p \cdot c_g + s \cdot c_f}{s \cdot c_f}&,\; \alpha_g = \alpha_f, \alpha_h > 0 \label{eq:limSpeedup1}\\
& (\Leftrightarrow 0=\alpha_g-\alpha_f < \alpha_h)\nonumber\\
\frac{s \cdot c_f}{s \cdot c_f + \frac{p \cdot c_g}{c_h}} + \frac{p \cdot c_g}{s \cdot c_f + \frac{p \cdot c_g}{c_h}} = \frac{s \cdot c_f + p \cdot c_g}{s \cdot c_f + \frac{p \cdot c_g}{c_h}}& ,\; \alpha_g = \alpha_f, \alpha_h = 0 \label{eq:limSpeedup2}\\
& (\Leftrightarrow 0=\alpha_g-\alpha_f = \alpha_h)\nonumber\\
1 + 0 = 1& ,\; \alpha_g < \alpha_f \label{eq:limSpeedup3}\\
& (\Leftrightarrow \alpha_g-\alpha_f < 0)\nonumber\\
\infty \;(\Theta (N^{\alpha_h})) & ,\; \alpha_g > \alpha_f, \alpha_h > 0, \alpha_g - \alpha_h - \alpha_f \geq 0 \label{eq:limSpeedup4}\\
& (\Leftrightarrow 0<\alpha_h\leq\alpha_g-\alpha_f)\nonumber\\
\infty \;(\Theta (N^{(\alpha_g - \alpha_f)})) & ,\; \alpha_g > \alpha_f, \alpha_h > 0, \alpha_g - \alpha_h - \alpha_f < 0 \label{eq:limSpeedup5}\\
& (\Leftrightarrow 0<\alpha_g-\alpha_f < \alpha_h)\nonumber\\
0 + c_h = c_h& ,\; \alpha_g > \alpha_f, \alpha_h = 0 \label{eq:limSpeedup6}\\
& (\Leftrightarrow 0=\alpha_h<\alpha_g-\alpha_f)\nonumber
\end{empheq}

I now briefly discuss each of the six equations and refer to these as \textit{speedup cases}; a visual illustration of the speedup cases can be retrieved from Figure \ref{fig:speedupLimits}. 

\begin{figure}
\fbox{\includegraphics[width=\linewidth]{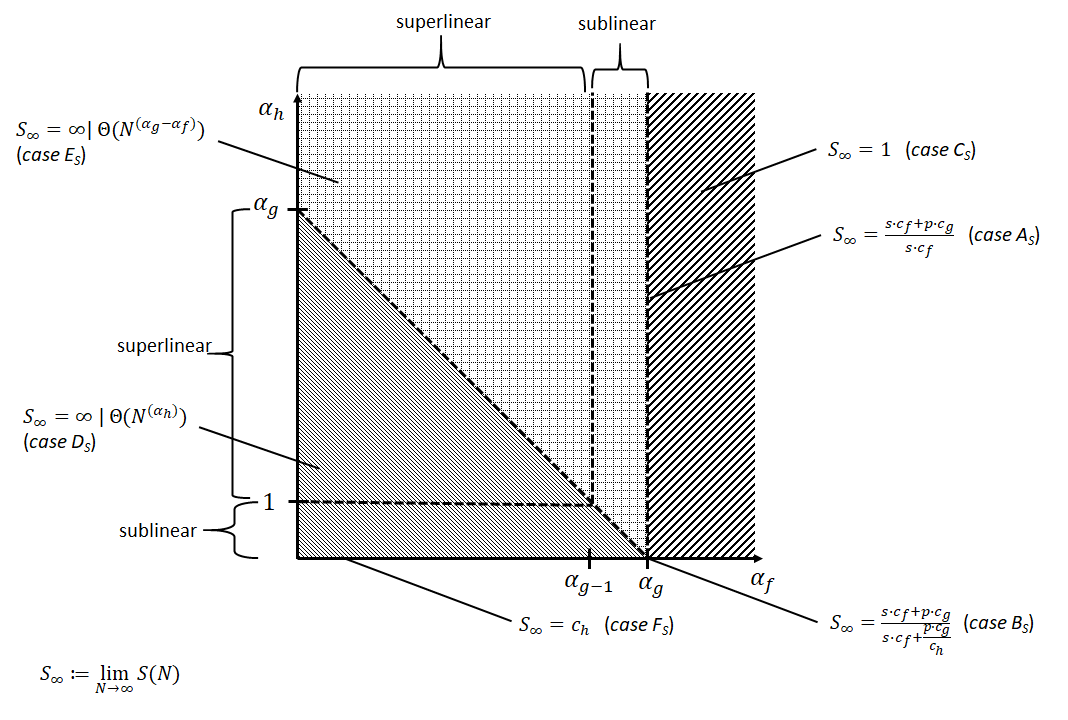}}
\caption{Overview of speedup limits}
\label{fig:speedupLimits}
\end{figure}

\begin{itemize}
  \item[Case $A_S$:] The speedup limit given in eq. (\ref{eq:limSpeedup1}) represents an upper bound for $S(N)\; \forall N > 1$ (see Appendix \ref{appendixSec:A}) and refers to situations in which the number $N$ of available PUs affects the time required to address the parallelizable workload (due to $\alpha_h >0$). While the speedup limit holds for any $c_h,\alpha_h >0$, it seems reasonable to assume that $h(N)=\pol{c_h}{\alpha_h} \geq 1$ holds as increasing the number of  PUs from $N=1$ to $N>1$ should not lead to an increase of time required to execute the parallelizable workload. However, the speedup limit in this case does not depend on the values of $c_h$ and $\alpha_h$. Also, this case assumes that the scaling functions $f(N)=\pol{c_f}{\alpha_f}$ and $g(N)=\pol{c_g}{\alpha_g}=\pol{c_g}{\alpha_f}$ change the serial and parallelizable workloads using the same factor $\pol{}{\alpha_f}$. It should be noticed that \textit{case $A_S$} results in \textit{Amdahl's law} \citep{amdahl1967validity} when setting $c_f=c_g=c_h=\alpha_h=1, \alpha_f=\alpha_g=0$; then, the limit on speedup amounts to $\frac{1}{s}$. As the speedup model of Amdahl's law is a special case of the memory-bound model suggested in \citep{sun1990another,sun1993scalable} (\textit{Sun and Ni's law}) under the assumption that the parallelizable workload in the memory-bound model is given by a power function $\pol{}{\alpha_g}$, \textit{case $A_S$} partly covers the abovementioned model. This case also (partly) covers the scaled workload model of \citet{schmidt2017parallel} under the assumption that the sequential and parallelizable workloads are given by power functions $\pol{c_f}{\alpha_f}$ and $\pol{c_g}{\alpha_g}$, resp., with $\alpha_f=\alpha_g$.
 \item[Case $B_S$:] The speedup limit given in eq. (\ref{eq:limSpeedup2}) represents an upper bound for $S(N)\; \forall N > 1$ (see Appendix \ref{appendixSec:A}). It refers to the same situation as described in \textit{case $A_S$} with the modification that, here, $N$ does not affect the time required to address the parallelizable workload (due to $\alpha_h =0$); that time is rather modified through a division by the scalar $c_h$; i.e., when executing the parallelizable workload in parallel, the corresponding time changes are determined by a constant factor $\alpha_h$. It seems reasonable to assume that $c_h \geq 1$ holds in this case (cmp. the analogeous discussion of \textit{case $A_S$}), with a resulting speedup limit of $\lim_{N \rightarrow \infty} S(N) \geq 1$. While \textit{case $B_S$} seems not useful under the premise that the parallelizable workload is infintely parallelizable, it becomes useful when this assumption is replaced by the expectation that a given parallelizable workload can be executed in parallel only on a limited number $N_{max}$ PUs; then, $\alpha_h$ may represent this limitation. For a discussion of limited parallelization, see, for example, \citep[p. 772ff]{cormen2022introduction} and \citep[p. 141ff]{al2020amdahl}).   

  \item[Case $C_S$:] When the increase of serial workload is asymptotically higher than that of parallelizable workload ($\alpha_f > \alpha_g$), speedup converges against $1$ (see eq. (\ref{eq:limSpeedup3})) as a lower bound regardless of the values of $c_h$ and $\alpha_h$; i.e., in this case, any parallelization does not reduce the overall execution time asymptotically due to the ``dominant"' increase of the serial workload. Case $C_S$ (partly) covers the scaled workload model of \citet{schmidt2017parallel} under the assumption that the sequential and parallelizable workloads are given by power functions $\pol{c_f}{\alpha_f}$ and $\pol{c_g}{\alpha_g}$, resp., with $\alpha_f<\alpha_g$.
   
  \item[Case $D_S$:] This case covers situations in which speedup is not limited and increases asymptotically with $\order{N^{\alpha_h}}$ for (i) $\alpha_g > \alpha_f$, (ii) $\alpha_h>0$, and (iii) $\alpha_g - \alpha_h \geq \alpha_f$ (see eq. (\ref{eq:limSpeedup4})); (i) ensures that the parallelizable workload increases faster than the sequential workload, (ii) ensures that parallelization asymptotically reduces the time required to execute the parallelizable workload, and (iii) ensures that the temporal effect induced through the joint growth of the parallelizable workload and its actual parallel execution ($\pol{}{(\alpha_g - \alpha_h)}$) is not weaker than the temporal effect induced through the growth of the sequential workload ($\pol{}{\alpha_f}$). Depending on the value of $\alpha_h$, speedup asymptotically grows sublinearly ($0 < \alpha_h < 1$), linearly ($\alpha_h = 1$), or superlinearly ($\alpha_h > 1$). It should be noticed that \textit{case $D_S$} results in \textit{Gustafson's law} \citep{gustafson1988reevaluating} when setting $c_f=c_g=c_h=\alpha_g=\alpha_h=1, \alpha_f=0$; then, the speedup asymptotically grows linearly. \textit{Case $D_S$} (partly) cover \textit{Sun and Ni's} law when setting $c_f=c_g=c_h=\alpha_h=1, \alpha_f = 0, \alpha_g \geq 1$ (under the assumption that the parallelizable workload in this model is given by a power function $\pol{}{\alpha_g}$). Finally, case $D_S$ (partly) covers the scaled workload model of \citet{schmidt2017parallel} under the assumption that the sequential and parallelizable workloads are given by power functions $\pol{c_f}{\alpha_f}$ and $\pol{c_g}{\alpha_g}$, resp., with $\alpha_g - \alpha_f \geq 1$.
    
  Interestingly, \textit{case $D_S$} may help explaining superlinear speedup as is has been observed in research on mathematical optimization (at least for a limited range of $N$)  \citep{ponz2017parallel,borisenko2011optimal,Rauchecker_Schryen_2019,Gonggiatgul2023}, for example.  
  
  \item[Case $E_S$:] Similar to \textit{case $D_S$}, \textit{case $E_S$} refers to situations in which speedup is not limited and increases asymptotically, but now with $\order{N^{(\alpha_g - \alpha_f)}}$ for (i) $\alpha_g > \alpha_f$, (ii) $\alpha_h>0$, and (iii) $\alpha_g - \alpha_h < \alpha_f$ (see eq. (\ref{eq:limSpeedup5})). The conditions under which \textit{case $E_S$} differ from those in \textit{case $D_S$} only with regard to (iii); i.e., here, the temporal effect induced through the joint growth of the parallelizable workload and its actual parallel execution ($\pol{}{(\alpha_g - \alpha_h)}$) is weaker than the temporal effect induced through the growth of the sequential workload ($\pol{}{\alpha_f}$). Now, the difference $(\alpha_g - \alpha_f)$ determines the asymptotic growth of speedup: it asymptotically grows sublinearly ($0 < \alpha_g - \alpha_f < 1$), linearly ($\alpha_g - \alpha_f = 1$), or superlinearly ($\alpha_g - \alpha_f > 1$).
  
  Similarly to \textit{case $D_S$}, \textit{case $E_S$} (partly) includes speedup models and laws suggested in the literature: \textit{Case $E_S$} (partly) covers \textit{Sun and Ni's law} when setting $c_f=c_g=c_h=\alpha_h=1, \alpha_f = 0, \alpha_g < 1$ (under the assumption that the parallelizable workload in this model is given by a power function $\pol{}{\alpha_g}$). With $\alpha_g=\frac{1}{2}$, Sun and Ni's model becomes the  ``generalized scaled speedup" model suggested in \citep{juurlink2012amdahl}; thus, \textit{case $E_S$} also covers the generalized scaled speedup model. Finally, case $E_S$ also (partly) includes the model of \citet{schmidt2017parallel} under the assumption that the sequential and parallelizable workloads are given by power functions $\pol{c_f}{\alpha_f}$ and $\pol{c_g}{\alpha_g}$, resp., with $0 < \alpha_g - \alpha_f < 1$); then, speedup grows asymptotically sublinearly with $\order{N^{(\alpha_g - \alpha_f)}}$.    

As \textit{case $D_S$}, also \textit{case $E_S$} may help explaining superlinear speedup.
 
  \item[Case $F_S$:] This case covers situations in which speedup converges to $c_h$ for (i) $\alpha_g > \alpha_f$ and (ii) $\alpha_h=0$ (see eq. (\ref{eq:limSpeedup6})). For $c_h \geq 1$, the limit $c_h$ is a lower bound; for $c_h \leq 1$, the limit $c_h$ is an upper bound. Condition (i) ensures that the parallelizable workload increases faster than the sequential workload, and with condition (ii) I assume that $N$ does not affect the time required to address the parallelizable workload (due to $\alpha_h =0$).  As with \textit{case $B_S$}, \textit{case $F_S$} is useful with the expectation that a given parallelizable workload can be executed in parallel only on a limited number $N_{max}$ of PUs. 
 
\end{itemize}  


\subsection{Asymptotic efficiency}
\label{sec:asymptoticEfficiency}

In order to determine theoretical efficiency limits, I proceed analogously to the determination of speedup limits. I rewrite the generic efficiency equation (eq. (\ref{eq:genericEfficiencyEquation})) as follows:
\begin{equation}
\label{eq:genericEfficiencyEquationRewritten}
 E(N)=\underbrace{\frac{\pol{s \cdot c_f}{\alpha_f} }{\pol{s \cdot c_f}{(\alpha_f + 1)} + \pol{\frac{p \cdot c_g}{c_h}}{(\alpha_g - \alpha_h + 1)}}}_{\substack{(I')}} + \underbrace{\frac{\pol{p \cdot c_g}{\alpha_g}}{\pol{s \cdot c_f}{(\alpha_f + 1)} + \pol{\frac{p \cdot c_g}{c_h}}{(\alpha_g - \alpha_h + 1)}}}_{\substack{(II')}}
\end{equation}

For term (I'), I yield the following limit (see equations (\ref{eq:I_eff_1}) - (\ref{eq:I_eff_3}) in Appendix \ref{appendixSec:B}):
\begin{equation}
\label{eq:I_eff}
\lim_{N\rightarrow \infty} (I')=0
\end{equation}

For term (II'), I yield the following limits (the proof can be obtained from equations (\ref{eq:II_eff_Def_a}) - (\ref{eq:II_eff_10}) in Appendix \ref{appendixSec:B}):
\begin{empheq}[left={\lim_{N\rightarrow \infty} (II')=}\empheqlbrace]{align}
0 &,\; \alpha_f > \alpha_g - 1 \label{eq:II_eff_Final1}\\
& (\Leftrightarrow \alpha_g-\alpha_f < 1)\nonumber\\
0 &,\; \alpha_f = \alpha_g - 1, 0 \leq \alpha_h < 1 \label{eq:II_eff_Final2}\\
& (\Leftrightarrow 0\leq \alpha_h < \alpha_g-\alpha_f = 1)\nonumber\\
\frac{p \cdot c_g}{s \cdot c_f + \frac{p \cdot c_g}{c_h}} &,\; \alpha_f = \alpha_g - 1, \alpha_h = 1 \label{eq:II_eff_Final3}\\
& (\Leftrightarrow \alpha_g-\alpha_f = \alpha_h = 1)\nonumber\\
\frac{p \cdot c_g}{s \cdot c_f} &,\; \alpha_f = \alpha_g - 1, \alpha_h > 1 \label{eq:II_eff_Final4}\\
& (\Leftrightarrow \alpha_g-\alpha_f = 1 < \alpha_h)\nonumber\\
0 &,\; \alpha_f < \alpha_g - 1, 0 \leq \alpha_h < 1 \label{eq:II_eff_Final5}\\
& (\Leftrightarrow 0<\alpha_h < 1 < \alpha_g-\alpha_f )\nonumber\\
c_h &,\; \alpha_f < \alpha_g - 1, \alpha_h = 1 \label{eq:II_eff_Final6}\\
& (\Leftrightarrow 1 = \alpha_h <\alpha_g-\alpha_f)\nonumber\\
\infty \;(\Theta (N^{\alpha_g - \alpha_f -1})) & ,\; \alpha_f < \alpha_g - 1, \alpha_h > 1, \alpha_f > \alpha_g - \alpha_h \label{eq:II_eff_Final7}\\
& (\Leftrightarrow 1<\alpha_g-\alpha_f < \alpha_h)\nonumber\\
\infty \;(\Theta (N^{\alpha_h -1})) & ,\; \alpha_f < \alpha_g - 1, \alpha_h > 1, \alpha_f \leq \alpha_g - \alpha_h \label{eq:II_eff_Final8}\\
& (\Leftrightarrow 1<\alpha_h \leq \alpha_g-\alpha_f)\nonumber
\end{empheq}

\newpage
With $\lim_{N\rightarrow \infty} (I')=0$, I yield the following limits for efficiency equations (\ref{eq:genericEfficiencyEquation}) and (\ref{eq:genericEfficiencyEquationRewritten}):
\begin{equation}
\label{eq:limEfficiency1}
\lim_{N \rightarrow \infty} E(N) = \lim_{N \rightarrow \infty} \left[(I') + (II')\right] = \lim_{N \rightarrow \infty} (II')
\end{equation}

I now briefly discuss each of the eight equations and refer to these as \textit{(efficiency) cases}; a visual illustration of the efficiency cases can be retrieved from Figure \ref{fig:efficiencyLimits} which, unsurprisingly, shows structural similarities with the visual representation of speedup limits due to the relationship between efficiency and speedup as given by $E(N)=\frac{S(N)}{N}$. 

\begin{figure}
\fbox{\includegraphics[width=\linewidth]{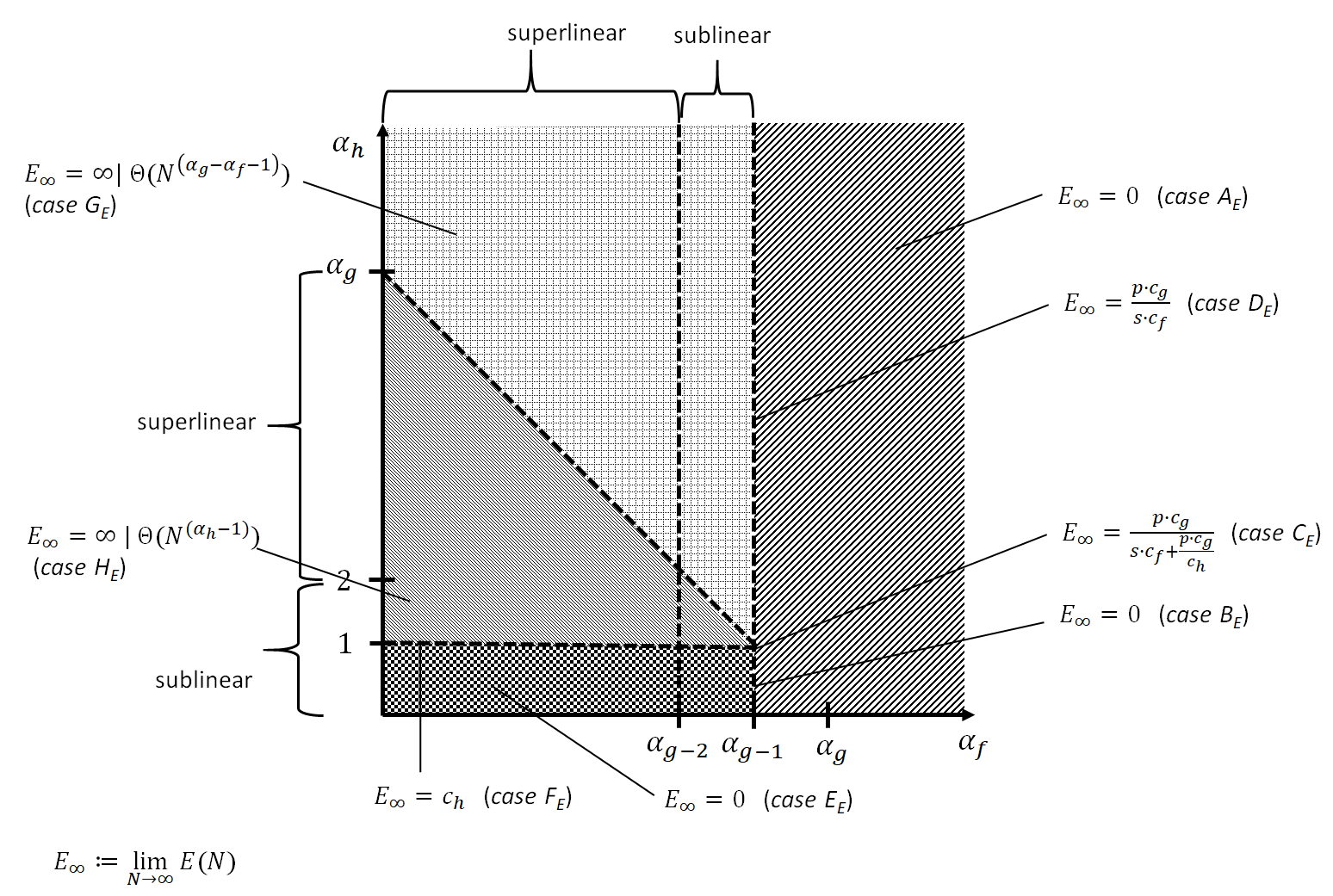}}
\caption{Overview of efficiency limits}
\label{fig:efficiencyLimits}
\end{figure}

\begin{itemize}
  \item[Case $A_E$:] The efficiency limit given in eq. (\ref{eq:II_eff_Final1}) equals zero regardless of the value of $\alpha_h$ and apparently represents a lower bound for $E(N)\; \forall N > 1$. This case refers to a situation in which the increase of the serial workload is asymptotically higher than that of a(n) (adjusted) parallelizable workload (adjusted based upon a decrease of the number of PUs by $1$) ($\alpha_f > \alpha_g - 1$).  
  
  \item[Case $B_E$:] The efficiency limit given in eq. (\ref{eq:II_eff_Final2}) equals zero when the value of $\alpha_h$ is upper bounded by $1$. Again, it apparently represents a lower bound for $E(N)\; \forall N > 1$. This case refers to a situation in which the ratio of the increases of serial workload and adjusted parallelizable workload converges against the constant $\frac{c_f}{c_g}$ and in which the increase of time reduction of executing the parallelizable workload is sublinear in $N$ ($\alpha_h < 1$). 
   
  \item[Case $C_E$:] The efficiency limit given in eq. (\ref{eq:II_eff_Final3}) describes a situation that differs from that in \textit{case $B_E$} only in that the increase of time reduction of executing the parallelizable workload is now linear in $N$ ($\alpha_h=1$). Then, the limit of efficiency is given by a constant larger than $0$ assuming that the parallelizable workload is positive ($p>0$).   
  
   \item[Case $D_E$:] The efficiency limit given in eq. (\ref{eq:II_eff_Final4}) describes a situation that differs from that in \textit{case $B_E$} only in that the increase of time reduction of executing the parallelizable workload is now superlinear in $N$ ($\alpha_h>1$). Despite this increase of time reduction and due to the relatively large increase of the serial workload compared to that of the parallelizable workload ($\alpha_f = \alpha_g - 1$), the limit of efficiency is still given by a constant (larger than $0$) assuming that the parallelizable workload is positive ($p>0$).
   
  \item[Case $E_E$:] The efficiency limit given in eq. (\ref{eq:II_eff_Final5}) describes a situation that is similar to that of \textit{case $B_E$}. While $\alpha_h < 1$ holds again, the (adjusted and the non-adjusted) parallelizable workload grows faster than the serial workload ($\alpha_f < \alpha_g -1$). However, as the increase of time reduction of executing the parallelizable workload is sublinear in $N$ ($\alpha_h<1$), efficiency converges against $0$. 
  
  \item[Case $F_E$:] In contrast to \textit{case $E_S$}, the efficiency limit given in eq. (\ref{eq:II_eff_Final6}) describes a situation in which the increase of time reduction of executing the parallelizable workload is linear in $N$ ($\alpha_h = 1$). Then, efficiency asymptotically equals a constant larger than $0$ (assuming $p > 0$).
   
 \item[Case $G_E$:] One situation in which efficiency is unbounded is described in eq. (\ref{eq:II_eff_Final7}), where the (adjusted and the non-adjusted) parallelizable workload grows faster than the serial workload ($\alpha_f < \alpha_g -1$) and the increase of time reduction of executing the parallelizable workload is superlinear in $N$ ($\alpha_h > 1$). When also $\alpha_f > \alpha_g - \alpha_h$ holds, efficiency grows asymptotically with $\Theta (N^{\alpha_g - \alpha_f -1})$; i.e., the asymptotic growth does not depend on $\alpha_h$. For $\alpha_g - \alpha_f > 1$, $\alpha_g - \alpha_f = 1$ and $\alpha_g - \alpha_f < 1$, efficiency is superlinear, linear and sublinear, respectively.

\item[Case $H_E$:] A second situation in which efficiency is unbounded is described in eq. (\ref{eq:II_eff_Final8}), where the (adjusted and the non-adjusted) parallelizable workload grows faster than the serial workload ($\alpha_f < \alpha_g -1$) and the increase of time reduction of executing the parallelizable workload is superlinear in $N$ ($\alpha_h > 1$). When also $\alpha_f \leq \alpha_g - \alpha_h$ holds, efficiency grows asymptotically with $\Theta (N^{\alpha_h - 1})$; i.e., the asymptotic growth does now depend on $\alpha_h$. For $\alpha_h > 1$, $\alpha_h = 1$ and $0 < \alpha_h < 1$, efficiency is superlinear, linear and sublinear, respectively.

\end{itemize}  

\subsection{Asymptotic scalability}
\label{sec:asymptoticScalability}

In the previous subsections, I identify \textit{speedup cases} and \textit{efficiency cases}. Now, I consider speedup cases and efficiency cases jointly, which results in various speedup-efficiency cases. I refer to these cases as \textit{scalability cases}, which are defined by both speedup and efficiency limits (see Table \ref{tab:scalabilityCases} and Figure \ref{fig:scalabilityLimits}). I assign to each scalability case a scalability type, which describes both speedup (as first parameter) and efficiency (as second parameter), using the following semantics:

\begin{itemize}
	\item $\beta_{h}, \gamma_{h}$: fixed values which depend on the reduced parallel workload scaling function $h$
    \item $\infty_{f,g}$: unbounded and monotonically increasing; the extent of increase depends on workload scaling functions $f$ and $g$
	\item $\beta_{s,f,g}, \gamma_{s,f,g}$: fixed values which depend on the sequential part $s$ (note: $p=1-s$) and the workload scaling functions $f$ and $g$
	\item $\beta_{s,f,g,h}, \gamma_{s,f,g,h}$: fixed values which depend on the sequential part $s$ (note: $p=1-s$), the workload scaling functions $f$ and $g$, and the reduced parallel workload scaling function $h$
	\item $\infty_{h}$: unbounded and monotonically increasing; the extent of increase depends on reduced parallel workload scaling function $h$
\end{itemize}  

Each scalability type refers to exactly one scalability case and set of conditions (see Table \ref{tab:scalabilityCases}).  

\begin{sidewaystable}
\center
\caption{Scalability cases}
\label{tab:scalabilityCases}
\begin{spacing}{1.2}
\begin{tabular}{|c|cc|ll|c|p{3cm}|}
\hline 
\multicolumn{3}{|c|}{Cases} & \multicolumn{2}{c|}{Limit values} & \multirow{2}{*}{Type} & \multirow{2}{*}{Conditions}\\ 
\cline{1-5} 
Scalability & Speedup & Efficiency & Speedup${}^{\ast}$ & Efficiency${}^{\ast \ast}$ & &\\ 
\hline 
$A_{SC}$ & $C_S $ & $A_E$ & $1$ & $0$ & $(1/0)$ & $\alpha_g -\alpha_f < 0$ \\ 
\hline 
$B_{SC}$ & $A_S$ & $A_E$ & $\beta:=\frac{s \cdot c_f + p \cdot c_g}{s \cdot c_f} > 1$ & $0$ & $(\beta_{ s,f,g}/0)$ & $0 = \alpha_g-\alpha_f<\alpha_h$  \\ 
\hline 
$C_{SC}$ & $B_S$ & $A_E$ & $\beta:=\frac{s \cdot c_f + p \cdot c_g}{s \cdot c_f + \frac{p \cdot c_g}{c_h}}$ & $0$ & $(\beta_{s,f,g,h}/0)$ & $0=\alpha_g-\alpha_f = \alpha_h$\\ 
\hline 
$D_{SC}$ & $F_S$ & $E_E$, $B_E$ & $\beta:=c_h$  & $0$ & $(\beta_h/0)$ & $0 = \alpha_h < \alpha_g - \alpha_f$\\ 
\hline 
$E_{SC}$ & \multirow{4}{*}{$D_S$} & $E_E$, $B_E$ & \multirow{4}{*}{$\infty \left( \Theta(N^{\alpha_h})\right)$} & $0$ & $(\infty_h/0)$ & $0 < \alpha_h < 1 \leq \alpha_g - \alpha_f$\\ 
\cline{1-1}\cline{3-3}\cline{5-6} 
$F_{SC}$ & & $H_E$& & $\infty \left( \Theta(N^{(\alpha_h - 1)})\right)$ & $(\infty_h/\infty_h)$ & $1 < \alpha_h \leq \alpha_g - \alpha_f$\\ 
\cline{1-1}\cline{3-3}\cline{5-6} 
$G_{SC}$ & & $C_E$ & & $\gamma:=\frac{p \cdot c_g}{s \cdot c_f + \frac{p \cdot c_g}{c_h}}$ & $(\infty_h/\gamma_{s,f,g,h})$ & $\alpha_h = \alpha_g - \alpha_f = 1$ \\ 
\cline{1-1}\cline{3-3}\cline{5-6} 
$H_{SC}$ & & $F_E$ & & $\gamma:=c_h$ & $(\infty_h/\gamma_h)$ & $ 1=\alpha_h < \alpha_g - \alpha_f$\\ 
\hline 
$I_{SC}$ & \multirow{3}{*}{$E_S$} & $G_E$ & \multirow{3}{*}{$\infty \left( \Theta(N^{(\alpha_g - \alpha_f)})\right)$} & $\infty \left( \Theta(N^{(\alpha_g - \alpha_f - 1)}\right)$ & $(\infty_{f,g}/\infty_{f,g})$ & $1 <  \alpha_g - \alpha_f < \alpha_h$ \\ 
\cline{1-1}\cline{3-3}\cline{5-6} 
$J_{SC}$ & & $A_E$& & $0$ & $(\infty_{f,g}/0)$ & $0 < \alpha_g - \alpha_f < \min{\left\{1,\alpha_h\right\}}$ \\ 
\cline{1-1}\cline{3-3}\cline{5-6}  
$K_{SC}$ & & $D_E$ & & $\gamma:=\frac{p \cdot c_g}{s \cdot c_f}$ & $(\infty_{f,g}/\gamma_{s,f,g})$ & $1 = \alpha_g - \alpha_f < \alpha_h$ \\ 
\hline
\multicolumn{7}{l}{\small ${}^{\ast}$: Values for cases $A_S$ and $B_S$ are upper bounds, value for case $C_S$ is lower bound, value for case $F_S$ is lower bound ($c_h \geq 1$) or upper bound ($c_h \leq 1$).}\\
\multicolumn{7}{l}{\small ${}^{\ast \ast}$: Values for cases $A_E$ to $G_E$ are lower bounds, the value for case $H_E$ is an upper bound (for sufficiently large values of $N$ ($c_h \cdot N \geq 1$)).}\\
\end{tabular} 
\end{spacing}
\end{sidewaystable}
\clearpage
  
\begin{figure}
\fbox{\includegraphics[width=\linewidth]{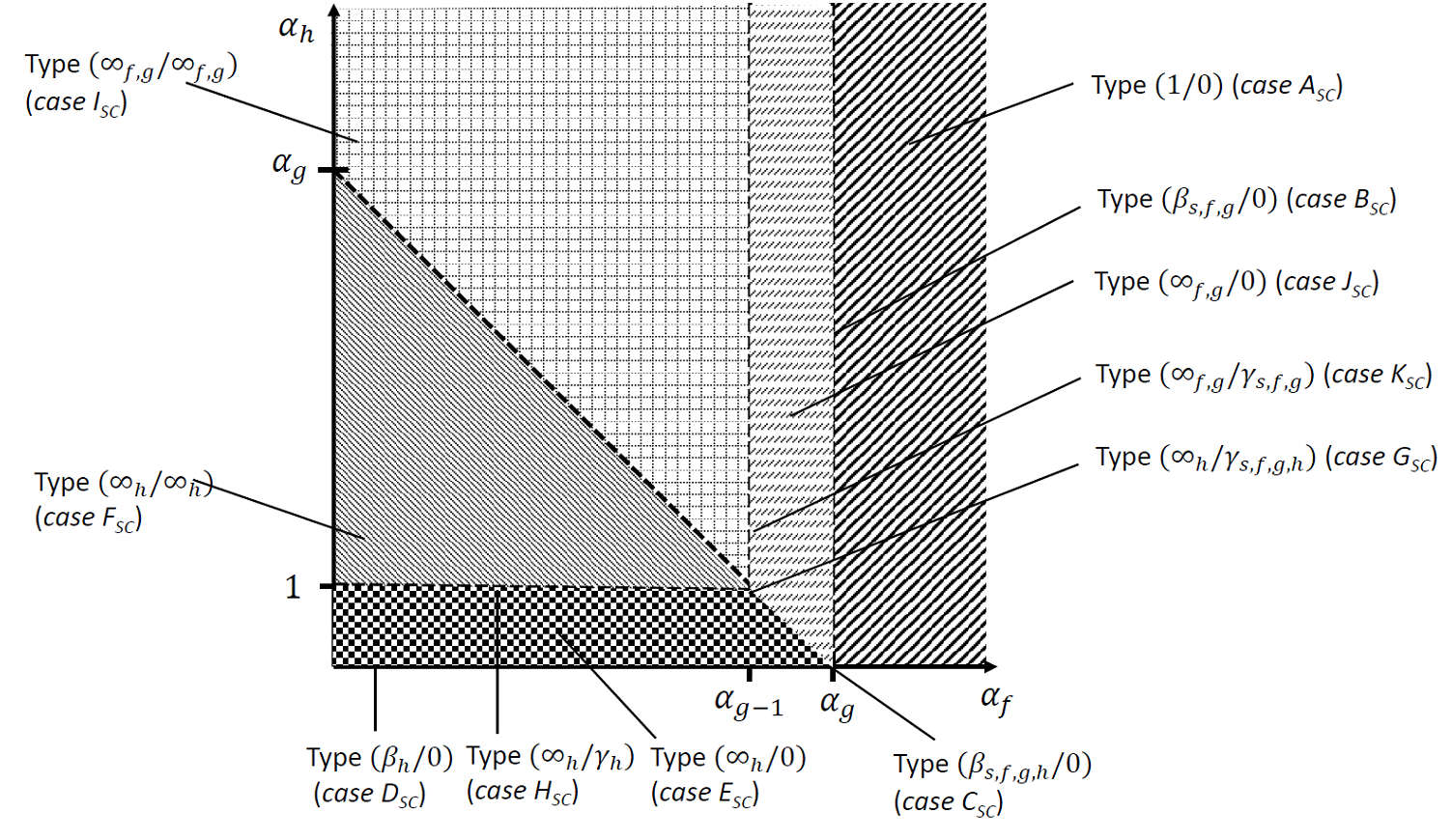}}
\caption{Overview of scalability cases}
\label{fig:scalabilityLimits}
\end{figure}

For the discussion of speedup, efficiency and scalability results derived in the preceding section, I recall the meaning of parameters $\alpha_f$, $\alpha_g$ and $\alpha_h$ since their values determine the (speedup, efficiency and scalability) case of a particular parallel algorithm: The parameters $\alpha_f$,  $\alpha_g$ and $\alpha_h$ affect the serial workload, the parallelizable workload and the actual reduction of parallelizable workload through parallelization, respectively, depending on the number of PUs $N$; they are given by $f(N):=c_f N^{\alpha_f}$, $g(N):=c_g N^{\alpha_g}$ and $h(N):=c_h N^{\alpha_h}$, respectively. I also recall the abovementioned speedup and efficiency equations:
\begin{displaymath}
S(N)=\frac{\pol{s \cdot c_f}{\alpha_f} + \pol{p \cdot c_g}{\alpha_g}}{\pol{s \cdot c_f}{\alpha_f} + \pol{\frac{p \cdot c_g}{c_h}}{(\alpha_g - \alpha_h)}}\; \text{(see eq. (\ref{eq:genericSpeedupEquation}))}
\end{displaymath}
\begin{displaymath}
E(N)=\frac{\pol{s \cdot c_f}{\alpha_f} + \pol{p \cdot c_g}{\alpha_g}}{\pol{s \cdot c_f}{(\alpha_f + 1)} + \pol{\frac{p \cdot c_g}{c_h}}{(\alpha_g - \alpha_h +1)}}\; \text{(see eq. (\ref{eq:genericEfficiencyEquation}))}
\end{displaymath}

I now discuss each of the eleven scalability cases $A_{SC}$ to $K_{SC}$. As the definition of scalability cases (and types) is based upon combinations of speedup cases and efficiency cases, the characteristics of scalability cases (and types) can be derived from the above descriptions of the characteristics of speedup and efficiency cases. 

\begin{itemize}
  \item[Case $A_{SC}$:] With speedup case $C_S$, the increase of serial workload is asymptotically higher than that of parallelizable workload ($\alpha_f > \alpha_g$). Then, speedup converges against $1$. Case $C_S$ induces efficiency case $A_E$ so that the resulting asymptotic efficiency equals $0$. The scalability type is $(1,0)$. Overall, parallelization does not scale at all and parallelization efforts do not make much sense.
  
  \item[Case $B_{SC}$:] With speedup case $A_S$, the scaling functions $f$ and $g$ change the serial and parallelizable workloads using the same factor $\pol{}{\alpha_f}$ with possibly different values $c_f$ and $c_g$; furthermore, case $A_S$ refers to situations in which the number $N$ of available PUs affects the time required to address the parallelizable workload (due to $\alpha_h >0$). Then, speedup converges against a constant $\beta_{s,f,g}=\frac{s \cdot c_f + (1-s) \cdot c_g}{s \cdot c_f}>1$. Speedup case $A_S$ leads to efficiency case $A_E$; i.e., efficiency converges against zero. 
  
  Scalability case $B_{SC}$, which refers to scalability type $(\beta_{s,f,g},0)$, includes \textit{Amdahl's law} \citep{amdahl1967validity} and partly \textit{Sun and Ni's law}  \citep{sun1990another,sun1993scalable} (see the discusssion of speedup case $A_S$).   
   
  \item[Case $C_{SC}$:] This scalability case is similar to the scalability case $B_{SC}$ and differs from it only as $\alpha_h$ equals zero; i.e., $N$ does not affect the time required to address the parallelizable workload. With speedup case $B_S$ and resulting efficiency case $A_E$, the associated scalability type is $(\beta_{s,f,g,h},0)$, with speedup $\beta_{s,f,g,h}=\frac{s \cdot c_f + (1-s) \cdot c_g}{s \cdot c_f + \frac{(1-p) \cdot c_g}{c_h}}$. As discussed above, speedup case $B_S$, and thus scalability case $C_{SC}$, are not useful under the premise that the parallelizable workload is infinitely parallelizable, but it becomes useful when a given parallelizable workload can be executed in parallel only on a limited number of PUs. 
  
   \item[Case $D_{SC}$:] Scalability case $D_{SC}$, which corresponds to scalability type $(\beta_h/0)$, includes speedup case $F_S$, in which the parallelizable workload increases faster than the sequential workload ($\alpha_g > \alpha_f$) and $N$ does not affect the time required to address the parallelizable workload ($\alpha_h =0$). Then, speedup converges to $c_h$. When speedup case $F_S$ applies, either efficiency case $E_E$ or $B_E$ applies with efficency converging towards zero in both cases. As with scalability case $C_{SC}$, case $D_{SC}$ is useful with the expectation that a given parallelizable workload can be executed in parallel only on a limited number of PUs.
   
  \item[Case $E_{SC}$:] This scalability case refers to a situation in which (i) the parallelizable workload increases at least one magnitude faster than the sequential workload ($\alpha_g \geq \alpha_f + 1$) and (ii) the number of available PUs affects the time required to address the parallelizable workload with unlimited and sublinear growth ($0 < \alpha_h < 1$). This scalability case is linked to speedup case $D_S$ and one of the efficiency cases $E_E$ or $B_E$, resulting to scalability type ($\infty_h,0$). Case $E_{SC}$ involves a speedup growth of $\Theta(N^{\alpha_h})$; i.e., speedup convergence is determined by the reduced parallel workload scaling function $h$. Due to the condition $\alpha_h < 1$, this growth is sublinearly in $N$ and efficiency converges towards zero.  
  
  \item[Case $F_{SC}$:] This case describes to a situation in which the parallelizable workload increases more than one magnitude faster than the sequential workload ($\alpha_g \geq \alpha_f + 1$) and the number of available PUs affects the time required to address the parallelizable workload with superlinear growth ($1 < \alpha_h \leq (\alpha_g - \alpha_f))$). Under such conditions, speedup case $D_S$ and efficiency case $H_E$ apply, resulting in the scalability type $(\infty_h,\infty_h)$, where both speedup and efficiency are unlimited, speedup grows superlinearly with $\Theta(N^{\alpha_h})$, and efficiency grows sublinearly (when $1 < \alpha_h < 2$), linearly (when $\alpha_h = 2$), or superlinearly (when $2 < \alpha_h$).   
  
 \item[Case $G_{SC}$:] This scalability case describes a situation in which the parallelizable workload increases one magnitude faster than the sequential workload ($\alpha_g = \alpha_f + 1$) and the number of available PUs affects the time required to address the parallelizable workload with linear growth $( \alpha_h = 1$). Under such conditions, speedup case $D_S$ and efficiency case $C_E$ apply, resulting in the scalability type $(\infty_h,\gamma_{s,f,g,h})$, where speedup is unlimited and grows linearly and where efficiency converges against a constant $\gamma$ that depends upon functions $f,g,h$ and upon $s$ ($\gamma_{s,f,g,h}=\frac{p \cdot c_g}{s \cdot c_f + \frac{p \cdot c_g}{c_h}}$).
 
Scalability case $G_{SC}$ covers \textit{Gustafson's law} \citep{gustafson1988reevaluating}, where $\alpha_g=\alpha_h=1, \alpha_f=0$. It also (partly) covers the model of \citet{schmidt2017parallel}.            

\item[Case $H_{SC}$:] Scalability case $H_{SC}$ differs from scalability case $G_{SC}$ only in the regard that the parallelizable workload increases more than one magnitude faster than the sequential workload ($\alpha_g > \alpha_f + 1$). Similiar to case $G_{SC}$, the scalability type is $(\infty_h,\gamma_h)$ but here $\gamma_h$ is set to $c_h$ (efficiency case $F_E$).

Analogously to scalability case $G_{SC}$, also case $H_{SC}$ (partly) covers the model of \citet{schmidt2017parallel}. In addition, case $H_{SC}$ also (partly) covers \textit{Sun and Ni's law} when $c_f=c_g=c_h=1$ holds and  when the parallelizable workload in this model is given by a power function $N^{\alpha_g}$.  

\item[Case $I_{SC}$:] This scalability case applies when (i) the parallelizable workload increases at least one magnitude faster than the sequential workload ($\alpha_g > \alpha_f + 1$) and (ii) the number of available PUs affects the time required to address the parallelizable workload with unlimited and superlinear growth ($1<\alpha_g-\alpha_f < \alpha_h$). Under these conditions, speedup case $E_S$ and efficiency case $G_E$ apply, leading to scalability type $(\infty_{f,g},\infty_{f,g})$; i.e., both speedup and efficiency are unbounded and grow superlinearly.   

\item[Case $J_{SC}$:] The conditions under which this case apply differ from those of case $I_{SC}$ in that the parallelizable workload increases less than one magnitude faster than the sequential workload  ($0 < \alpha_g - \alpha_f < 1$). Then, speedup case $E_S$ and efficiency case $A_E$ apply, leading to scalability type $(\infty_{f,g},0)$; i.e., while speedup is unbounded and grow sublinearly, efficiency converges against zero.  

With $\alpha_f=0, \alpha_h=1$, it partly covers \textit{Sun and Ni's law}, and with $\alpha_f=0, \alpha_g=\frac{1}{2},\alpha_h=1$, it covers the speedup model of \citet{juurlink2012amdahl}.  

\item[Case $K_{SC}$:] When (i) the parallelizable workload increases one magnitude faster than the sequential workload ($\alpha_g = \alpha_f + 1$) and (ii) the number of available PUs affects the time required to address the parallelizable workload with unlimited and superlinear growth ($1< \alpha_h$), scalability case $K_{SC}$ applies with speedup case $E_S$ and efficiency case $D_E$; then, speedup is unlimited and grows linearly, efficiency converges against a constant $\gamma=\frac{p \cdot c_g}{s \cdot c_f}$, and scalability type $(\infty_{f,g},\gamma_{s,f,g})$ applies.  

\end{itemize}  

\section{Computational experiments}
\label{sec:computationalExperiments}

\color{red}

I demonstrate the application of the generic speedup and efficiency model and the scalability typology with two computational experiments. The first experiment targets strong scalability and performs parallel matrix multiplication with fixed workload (in terms of matrix sizes) in the Amdahl setting. The second experiment is more sophisticated and targets lower-upper (LU) decomposition as factorization of a matrix as the product of a lower triangular matrix and an upper triangular matrix; here I consider a variable workload (scaled-size model) that increases with the number of available PPUs. Both types of tasks occur in many problems in numerical analysis and linear algebra. The description of these experiments is intended to illustrate the proposed model and to give examples of how the polynomial functions of the generic speedup and efficiency model can be determined in practice.  

I ran all experiments on a compute node (2x AMD Milan 7763, 2.45 GHz, up to 3.5 GHz, 2x 64 cores, 256 GiB main memory) of the HPC cluster NOCTUA2 provided by the Paderborn Center for Parallel Computing of Paderborn University\footnote{\url{https://pc2.uni-paderborn.de/hpc-services/available-systems/noctua2}}. The source code was written in C++, using the OpenMP application programming interface for parallelization, and compiled with GCC (version 12.2.0)\footnote{I used the following compiler options: \textit{-march=native -m64 -fPIC -O3 -fopenmp}}.

\subsection{Parallel matrix multiplication with fixed workload}
\label{sec:ParMatMultFixedWork}

I perform parallel matrix multiplication without any optimization, i.e. rows and columns of the first and second matrices are multiplied in pairs (scalar product).
In my speedup and efficiency model I set $c_f=c_g=c_h=\alpha_h=1, \alpha_f=\alpha_g=0$. This setting results in the following speedup and efficiency equations:
\begin{equation}
\label{eq:example1}
   S(N)=\frac{1}{s+\frac{p}{N}},\quad E(N)=\frac{1}{N\cdot s + p}
\end{equation}

This instantiation of the generic speedup and efficiency model corresponds to scalability class $B_{SC}$ and scalability type $(\beta_{s,f,g}=1/s,0)$ (see Table \ref{tab:scalabilityCases}); i.e., asymptotic speedup and asymptotic efficiency are given by $1/s$ and $0$, respectively.
I set the sizes of the two integer matrices $A$ and $B$ to $[1.28 \cdot 10^{10},300]$ and $[200,1.28 \cdot 10^{10}]$, respectively, and initialized both matrices with random values in the interval $[-1000,1000]$. The large sizes of the matrices are motivated by the fact that parallelization may then be useful. I also set the number of rows to a multiple of the maximum number of cores available on the target machine ($128$), to allow an equal distribution of the parallelizable workload when $128$ (or lower powers of $2$) cores are used. This justifies setting $h(N):=N$. I determined the sequential and parallelizable workloads $s=0.023595$ and $p=0.976405$, respectively, by executing the code on a single core. Not surprisingly, in a matrix multiplication setting, the serial workload is quite small. Note, however, that even for a value as low as $s=0.023595$, the speedup is limited by $1/s\approx 42$.

\begin{figure}[pos=htb]
                    \centering
                        \begin{tabular}{cc}
                            (a) Speedup  & (b) Efficiency\\
                            \includegraphics[width=0.45\textwidth]{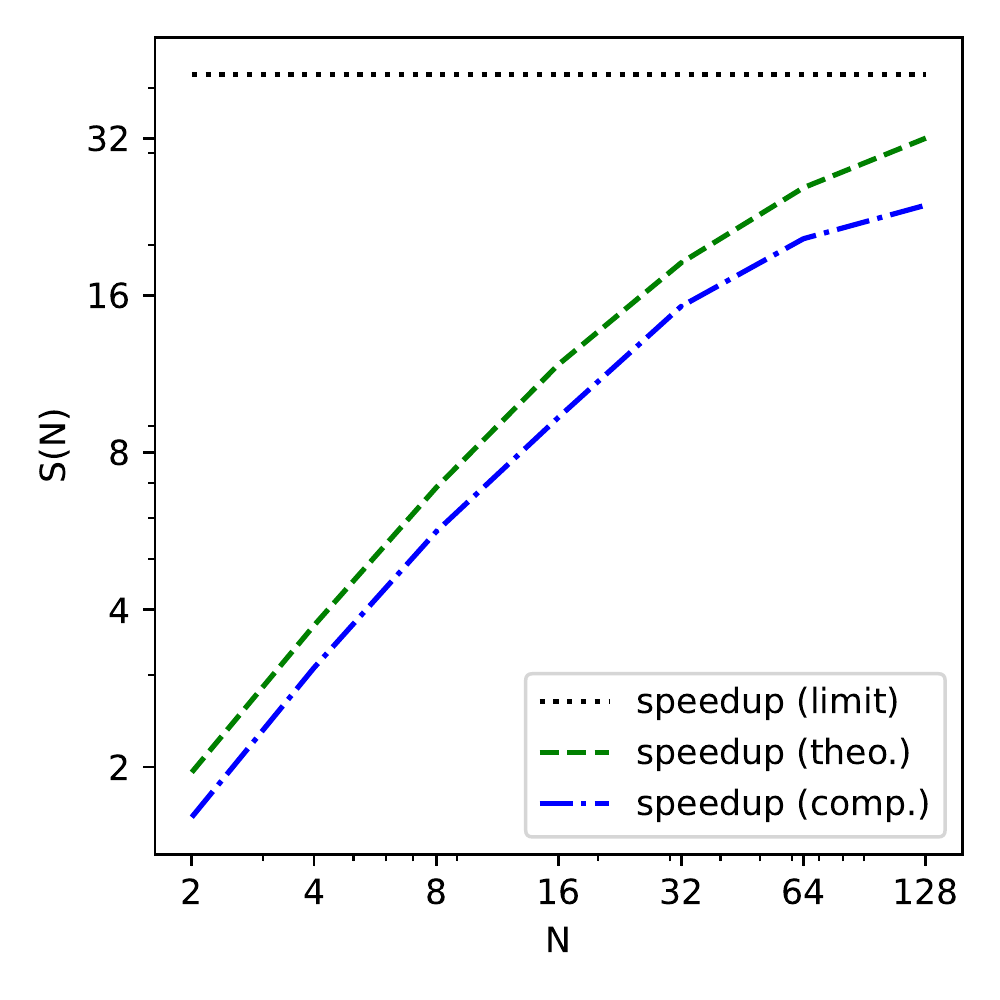}
                            & 
                            \includegraphics[width=0.51\textwidth]{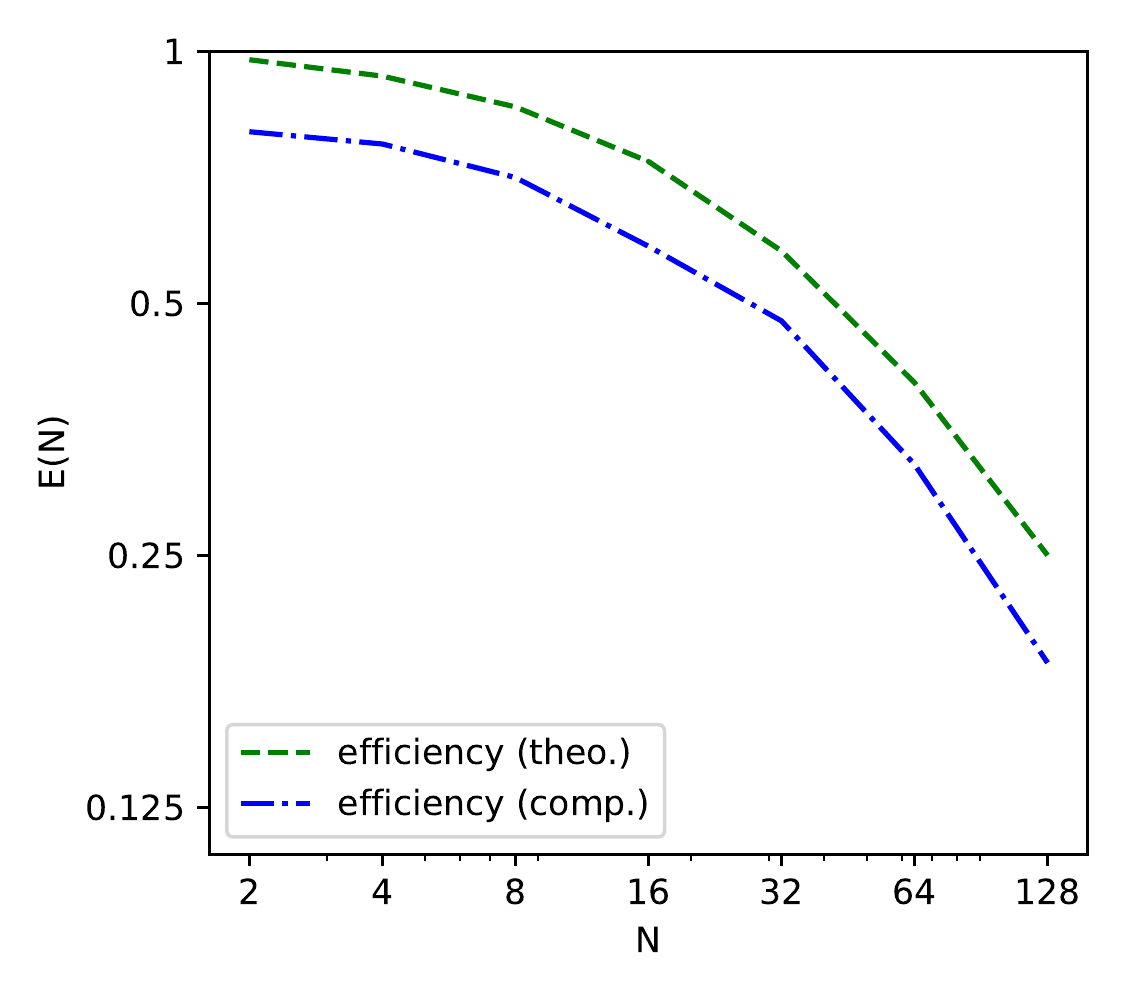}
                        \end{tabular}
                    \caption{Computational results of matrix multiplication with a fixed workload}
                    \label{fig:matMultFixedWork}
                \end{figure}

Figure \ref{fig:matMultFixedWork}a shows that the theoretical speedup, which considers the existence of a serial part $s>0$ but ignores any computational overhead induced by parallelization, converges to the speedup limit given by $1/s\approx 42$. The computational speedup, given by $T(1)/T(N)$ for any $N$, shows the actual speedup achieved. The gap between the two speedup lines is due to the overhead of parallelization. Figure \ref{fig:matMultFixedWork}b shows both the theoretical and the computational efficiency, defined as the ratios of the above speedup values and the number of parallel PUs $N$. While the decrease of the theoretical efficiency and its convergence to $0$ is due to the presence of a serial part $s>0$, the gap between the two efficiency curves is, as in the speedup case, due to the computational overhead induced by parallelization. The data for this experiment can be found in Table \ref{tab:resultsMatrixMultFixedSize} in Appendix \ref{sec:AppendixExperiment1}. 

\subsection{Parallel LU decomposition with variable workload}
\label{sec:ParLUDecomp}

I perform LU decomposition of an matrix without pivoting.  A description of the algorithm is shown in Figure \ref{fig:AlgLUDecomp}. The LU decomposition algorithm is only partially parallel. The outer loop (iteration over $i$) cannot be parallelized, since computations in iteration $i\;(i > 1)$ require computations in iteration $(i-1)$ to be completed. In contrast, the inner loops (iterating over $j$ and $l$) allow parallelization. In my implementation, the first inner loop (iteration over $j$) is parallelized. To account for variable workload, I assume that the number of rows in the matrix $A$ grows linearly with the number of PPUs.\footnote{Obviously, the size of $A$ in terms of the number of entries grows quadratically with the number of available PPUs.}.

\begin{figure}[htbp]
\centering
\fbox{
\begin{minipage}{\textwidth/3}
\begin{flushleft}
\begin{align*}
1. \quad & \textit{Input: } A \in \mathbb{N}^{z \times z},\,\mathbb{N}:=\{1,2,\ldots\} \\
2. \quad & L \leftarrow \text{identity matrix } I_n  \\
3. \quad & U \leftarrow A \\
4. \quad & \text{for } i = 1 \text{ to } z-1 \text{ do} \\
5. \quad & \quad \text{for } j = i+1 \text{ to } z \text{ do} \\
6. \quad & \quad\quad L_{j,i} \leftarrow \frac{U_{j,i}}{U_{i,i}} \\
7. \quad & \quad\quad \text{for } l = i+1 \text{ to } z \text{ do} \\
8. \quad & \quad\quad\quad U_{j,l} \leftarrow U_{j,l} - L_{j,i} \times U_{i,l} \\
9. \quad & \quad\quad \text{end for} \\
10. \quad & \quad\quad \text{for } l = 1 \text{ to } i \text{ do} \\
11. \quad & \quad\quad\quad U_{j,l} \leftarrow 0 \\
12. \quad & \quad\quad \text{end for} \\
13. \quad & \quad \text{end for} \\
14. \quad & \text{end for}\\
15. \quad & \textit{Output: } L,U \in \mathbb{N}^{z \times z},\, A=L \cdot U
\end{align*}
\end{flushleft}
\end{minipage}
}
\caption{LU decomposition algorithm (without pivoting)}
\label{fig:AlgLUDecomp}
\end{figure}

As in the first experiment, I assume that the serial workload remains constant as the number of PPUs varies; i.e., I assume $f(N)=1\, \; (N \geq 1)$. However, unlike the setting in the first experiment, the workload size now depends on the number of parallel PPUs $N$; i.e., one must determine an appropriate polynomial function $g(N)=c_g \cdot N^{\alpha_g}$. It is also reasonable to assume that the parallelizable workload $p \cdot g(N)$ cannot be partitioned into $N$ independent sub-workloads of equal size as the number of iterations in the $j$-loop is usually not a multiple of the number of available PPUs; i.e., the workload is not ``perfectly'' parallelizable and one needs to determine an appropriate polynomial function $h(N)=c_h \cdot N^{\alpha_h}$. I use analytical and numerical approaches to determine $g$ and $h$.

To determine $g$, I first determine the number of (basic) calculations $\hat{g}(z)$ (as performed in lines 6, 8, and 11 of the LU decomposition algorithm shown above) as a function of the number of matrix rows $z$. As shown in Appendix \ref{sec:AppendixExperiment2}, one obtains
\begin{equation}
\label{eq:numberCalcLUDecomp}
\hat{g}(z)= \frac{1}{3} \left( z^3-z \right)
\end{equation}
Assuming that a calculation requires $t_c$ units of computation time, the execution of (lines 4-14 of) the LU decomposition algorithm requires $\hat{g}(z) \cdot t_c$ units of computation time. As I assume that the number of rows/columns $z$ of the matrix $A$ (``problem size'') grows linearly with the number of PPUs $N$, $z$ is given by $z=z_1 \cdot N$, where $z_1$ is the number of rows when only a single processing unit is available. This leads to
\begin{equation}
\label{eq:numberCalcLUDecompN}
\hat{g}(z)=\hat{g}(N,z_1) =\hat{g}(z_1 \cdot N) = \frac{1}{3} \left( \left(z_{1}\cdot N)^{3}-(z_{1}\cdot N \right) \right) = \frac{1}{3} \left( z_{1}^{3} \cdot N^3 -z_1 \cdot N\right)
\end{equation}
For $N=1$ we get
\begin{equation}
\label{eq:numberCalcLUDecomp1}
\hat{g}(1,z_1)= \frac{1}{3} \left( z_{1}^{3} - z_1 \right),
\end{equation}
which is the number of calculations when only a single processor is available (i.e., $N=1$).  I can now determine a normalized workload scaling function $\overline{g}_{z_1}$ given by 
\begin{equation}
\label{eq:definitiongTempLUDecomp}
\overline{g}_{z_1}(N):= \frac{\hat{g}(N,z_1)}{\hat{g}(1,z_1)}
= \frac{z_{1}^{3} \cdot N^3 - z_1\cdot N}{z_{1}^{3} - z_1}
= \frac{z_{1}^{3}}{z_{1}^{3} - z_1} \cdot N^3 - \frac{z_{1}^{3}}{z_{1}^{3} - z_1} \cdot N = \Theta(N^3)
\end{equation}
Note that $\overline{g}$ is normalized as $\overline{g}_{z_1}(1)=1$. For the analysis of the asymptotic speedup, I do not need to consider the linear term any further, so I use as workload scaling function
\begin{equation}
\label{eq:definitiongLUDecomp}
g_{z_1}(N)=\frac{z_{1}^{3}}{z_{1}^{3} - z_1} \cdot N^3\end{equation}
Using the generic representation of the parallel workload scaling function $g(N)=c_g \cdot N^{\alpha_g}$, I get $c_g=z_1^3/ \left( z_1^3-z_1 \right)$ and $\alpha_g=3$. 

I now determine the scaling function $h$, which describes the extent to which the execution time of the parallelizable workload $p \cdot g_{z_1}(N)$ can theoretically be reduced by parallelization.  While the outer loop cannot be parallelized (see explanation above), the first inner loop (iteration over variable $j$) can be parallelized because the computations in the iterations do not affect each other. Assuming that this loop is parallelized, $(z-i)$ iterations can be executed in parallel, for each $i=1, \ldots, (z-1)$. Suppose $z=128$ and $N=64$ PPUs are available. Then, for example, in the $63$th iteration of the outer loop ($i=63)$, $(z-i)=65$ iterations of the first inner loop ($j=64, \ldots, 128$) can be distributed to $N=64$ PPUs. While $64$ of the $65$ iterations can be distributed evenly across $N=64$ PPUs, the remaining iteration is executed on a single PPU while the other $63$ PPUs remain idle. Obviously, the time required to execute the workload of this iteration of the outer loop is not divided by the number of available PPUs $N=64$. Therefore, $h(N)<N\;(N > 1)$.

The proof of equation (\ref{eq:numberCalcLUDecomp}) (see Appendix \ref{sec:AppendixExperiment2}) makes use of the fact that the number of calculations of the parallelizable part is given by
\begin{equation}
\label{eq:numberCalcLUDecompDecomposed}
\hat{g}(z)= \sum_{i=1}^{z-1} \left[ (z-i) \cdot (z-i+1) \right].
\end{equation}
Note that parallelizing the first inner loop means that $(z-i)$ blocks (of $(z-i+1)$ calculations each) can be executed in parallel by $N$ PPUs. Following my explanations of parallelization above, using $N$ PPUs requires a computation time equal to the computation time required in a serial execution for $\hat{g}_{reduced}(z)$ calculations, with
\begin{equation}
\label{eq:reducedg}
\hat{g}_{reduced}(z)= \sum_{i=1}^{z-1} \left[ \left\lceil \frac{(z-i)}{N} \right\rceil \cdot (z-i+1) \right].
\end{equation}
Assuming again that $z=N \cdot z_1$, then the computation time of the parallelizable workload executed on a single PPU is reduced by using $N$ PPUs by the factor $\hat{h}(z)=\hat{h}(N,z_1)$ given by
\begin{equation}
\label{equation}
\hat{h}(N,z_1) = \hat{h}(z) = \frac{\hat{g}(z)}{\hat{g}_{reduced}(z)} = \frac{\hat{g}(N,z_1)}{\hat{g}_{reduced}(N,z_1)}.
\end{equation}
Since it is difficult to find a closed-form expression of $\hat{g}_{reduced}$ and to derive the function $\hat{h}$ from it analytically, one could use a statistical approach to derive, for a given $z_1$, a polynomial function $h(N,z_1)=c_h \cdot N^{\alpha_h}\,(c_h>0,\alpha_h \geq 0)$ that approximates $\hat{h}$. For $z_1 \in \{10^i,\, i\in \{2,3,4,5\}\}$, Tables \ref{tab:dataPointsLRDecomp100}-\ref{tab:dataPointsLRDecomp100000} in Appendix \ref{sec:AppendixExperiment2} show data points for $N =2^i,\, i \in \{1,\ldots 20\}$. Note that $z_1$ is the number of rows/columns of the matrix to be decomposed when only a single PU is available, and that $c_h$ and $\alpha_h$ must be determined depending on $z_1$.

The computations show that for any value of $z_1$, the value of $\hat{h}(N,z_1)$ is close to $N$. It is reasonable to assume that, for any $z_1$, as $N$ goes to infinity, $(\hat{h}(N,z_1)-N)$ converges to $0$, or $(\hat{h}(N,z_1)/N)$ converges to 1, although I do not prove this here. However, from an asymptotic point of view, for any $z_1$ and any $c_h>0, \alpha_h<1$, $h(N,z_1)=c_h \cdot N^{\alpha_h}$ would converge to $0$. Therefore, I do not need to do a statistical analysis here and set $c_h=\alpha_h=1$; i.e., $h(N,z_1)=h(N)=N$. 

In summary, the scaling functions for the given LU decomposition algorithm are given by
\begin{align}
\label{eq:scalingFunctionsLUDecomp}
&f(N) = 1 &\qquad& (c_f=1,\alpha_f=0) \\
&g_{z_1}(N) = \frac{z_1^3}{z_1^3-z_1}\cdot N^3 && \left(c_g=\frac{z_1^3}{z_1^3-z_1}, \alpha_g=3\right) \notag \\
&h(N)=N && (c_h=\alpha_h=1), \notag
\end{align}

which yields as speedup equation (cmp. equation (\ref{eq:genericSpeedupEquation}))
\begin{equation}
\label{eq:SpeedupEquationLUDecomp}
S_{z_1}(N)=\frac{s+p\cdot \frac{z_1^3}{z_1^3-z_1}\cdot N^3}{s+p\cdot \frac{z_1^3}{z_1^3-z_1}\cdot N^2}
\end{equation}

Based on the proposed scalability typology shown in Table \ref{tab:scalabilityCases}, the LU decomposition algorithm is of scalability case $H_{SC}$ with scalability type $(\infty_h,\gamma_h)$; the speedup increases asymptotically with $\Theta(N)$ ($\lim_{N \rightarrow \infty} (S_{z_1}(N)-N)=0$) and the efficiency converges to $c_h=1$. Note that these results do not depend on the value of $z_1$.  

In the computational experiment, I initialize the matrix to be decomposed with (random and integer) values in the interval $[1,1000]$ to avoid any numerical problems, since the LU decomposition does not involve any pivoting. I set $z_1:=100$; thus, the number of rows and columns of the matrix to be decomposed is given by $z=N \cdot z_1=100\cdot N\,(N=1,2,4,\ldots, 128)$. The first experiments showed that, contrary to the theoretical analysis, the speedup tends to converge to $1$ with increasing values of $N$. This effect is due to the phenomenon that the amount of work done in the parallel region (lines 6-12 in the pseudocode shown in Figure \ref{fig:AlgLUDecomp}) is small relative to the overhead of creating and managing threads. To avoid this effect, I extended the LU decomposition algorithm to decompose a set of $m$ matrices $A_d\, (d=1,\ldots ,m)$ of equal size; this can be easily implemented by running the code in lines 6-12 for each of the $m$ matrices, so that the workload in the parallel region increases by a factor of $m$. Note that from a theoretical point of view, this modification does not affect the speedup and efficiency bounds, while it allows to demonstrate the speedup increase with increasing values of $N$ in computational experiments.  A second problem arose when determining the sequential and parallelizable fractions of the total execution time $s$ and $p$, respectively: since the matrix size (in terms of the number of rows/columns) grows linearly with $N$, and the number of computations to be performed in the LU decomposition algorithm grows asymptotically with $\Theta(n^3)$ (see equation (\ref{eq:definitiongTempLUDecomp})), the matrices to be decomposed become large. When only a single PU is used, the matrix to be decomposed must be relatively small in order to perform the experiments in a reasonable amount of time. As a consequence, the total execution time for decomposing a $(100 \times 100)$ matrix is small ($11$ ms), so the serial time $s$ is $0$ when traced computationally, since the numerical accuracy is too low. However, when $s$ is close to $0$, it hardly affects the theoretical (serial and parallel) computation times and speedups (see equation (\ref{eq:SpeedupEquationLUDecomp})); it also does not affect the speedup and efficiency bounds. Since in every parallel program at least a small part of the code is executed serially, I set $s:=0.01$.

\begin{figure}[pos=htb]
                    \centering
                        \begin{tabular}{cc}
                            (a) Speedup  & (b) Efficiency\\
                            \includegraphics[width=0.45\textwidth]{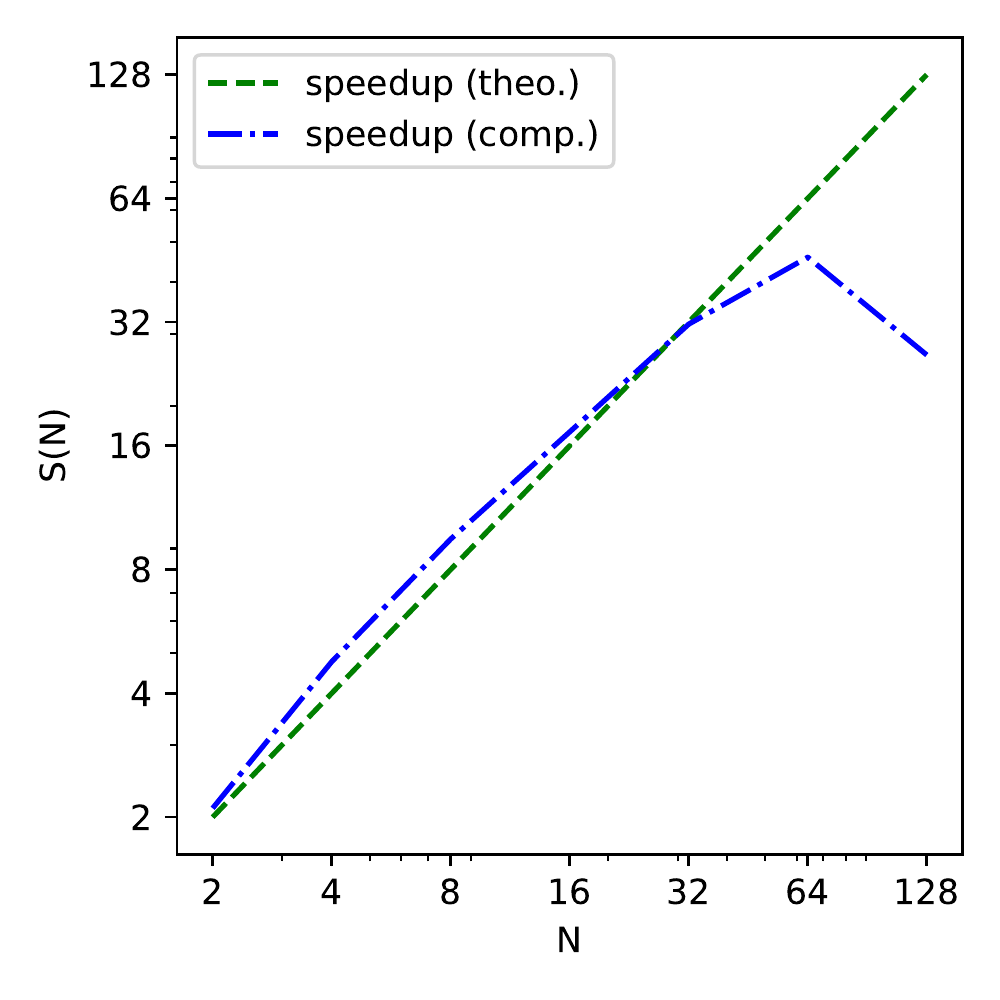}
                            & 
                            \includegraphics[width=0.51\textwidth]{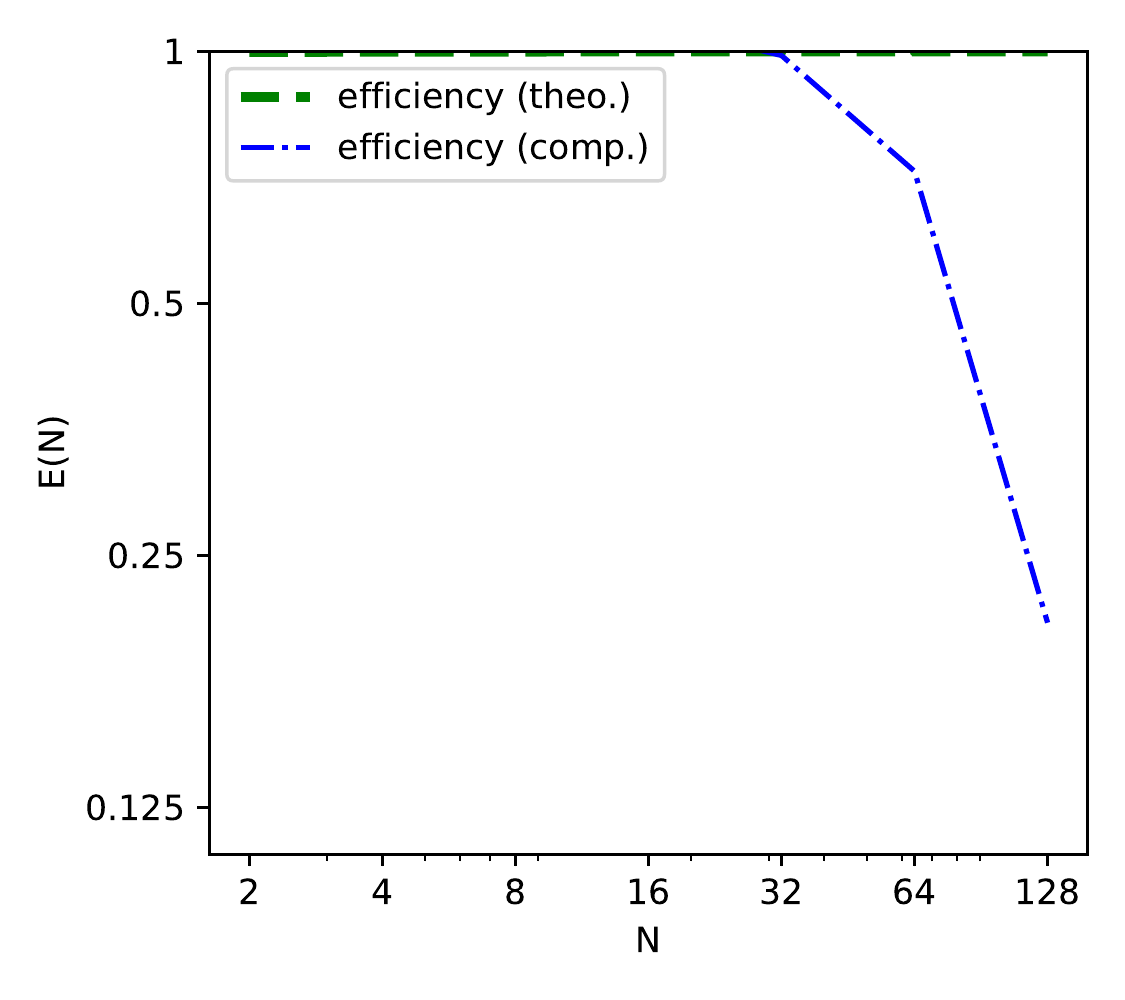}
                        \end{tabular}
                    \caption{Computational results of LU matrix decomposition with variable workload}
                    \label{fig:matLUDecompVarWork}
                \end{figure}

Figure \ref{fig:matLUDecompVarWork} shows the speedup and efficiency achieved in the computational experiment. As suggested by the theoretical analysis, the speedup increases linearly and the efficiency is close to $1$ (using up to $N=32$ threads). Since both the serial and parallel execution times are relatively small (see the figures in Table \ref{tab:resultsMatrixLUDecompVarSize} in  Appendix \ref{sec:AppendixExperiment2}), the computational speedup seems to be slightly above the theoretical speedup (for $N=2,4,8,16$) due to numerical problems in determining execution times on the cluster with sufficient accuracy. For $N=64$ the speedup still increases, but at a much lower rate (speedup $\approx$ 46) with an efficiency of about $0.72$; for $N=128$ the speedup even decreases (speedup $\approx$ 26.61) with an efficiency of about $0.21$. It seems reasonable to expect that for $N>128$ both speedup and efficiency decrease further. I did not investigate these effects because the computing node available for the experiments has a maximum of 128 cores. Obviously, and as expected, for an empirically large number of PPUs, theoretical speedup and efficiency bounds become much looser for computations. Therefore, in the next section, I discuss how the impact of parallelization overhead can be taken into account in future research when looking for tighter bounds.

\color{black}
\section{Discussion}
\label{sec:discussion}

\subsection{Application of model and typology}
\label{sec:findagoodheading}

The scalability typology developed in the previous section allows researchers to determine the limits of speedup and efficiency of their applications and the extent to which computational parallelization scales for their needs. They also support researchers regarding their decision of how many parallel PUs to use in the presence of economic budget constraints. My typology shown in Table \ref{tab:scalabilityCases} provides a more comprehensive picture of scalability in homogeneous computing environments than speedup laws suggested in the literature (shown in Table \ref{tab:generalizationSpeedupModelsLiterature}), thereby widening the scope of applying scalability insights. At the same time, my typology includes all of the abovementioned speedup laws as illustrated in Table \ref{tab:scalabilizyTypesSpeedupModelsLiterature}. 
\color{red}
In particular, Amdahl's law and Gustafson's law are consistent with our classification. These laws have been discussed in the literature as two different types of scaling: \emph{strong scaling} focuses on the ``Amdahl setting'', where the total problem size remains fixed as more processors are added, and the goal is to run the same problem size faster. In contrast, \emph{weak scaling} focuses on the ``Gustafson setting'', where the problem size per processor remains fixed as more processors are added, the total problem size is linear in the number of processors used, and the goal is to run larger problems in the same amount of time \citep{barney2010introduction}.
\color{black}

\begin{table}
\center
\caption{Scalability types of speedup models suggested in the literature}
\label{tab:scalabilizyTypesSpeedupModelsLiterature}\begin{spacing}{1.2}
\begin{tabular}{|p{3cm}|c|l|l|}
\hline
Speedup model & Scalability case & Scalability type & Conditions\\
\hline
 \textit{Amdahl's law} \citep{amdahl1967validity} & $B_{SC}$ & $(\beta_{s,f,g}=\frac{1}{s},0)$ & $c_f=c_g=c_h=\alpha_h=1, \alpha_f=\alpha_g=0$ \\[5mm]\hline
 \textit{Gustafson's law} \citep{gustafson1988reevaluating} & $G_{SC}$ & $(\infty_h,\gamma_{s,f,g,h}=1-s)$ &  $c_f=c_g=c_h=\alpha_g=\alpha_h=1, \alpha_f=0$ \\[5mm]\hline
 \textit{Generalized scaled speedup model} \citep{juurlink2012amdahl} & $J_{SC}$ & $(\infty_{f,g},0)$ & $c_f=c_g=c_h=\alpha_h=1, \alpha_f = 0, \alpha_g=\frac{1}{2}$ \\[5mm]\hline
 \textit{Sun and Ni's law} \citep{sun1990another,sun1993scalable} & $B_{SC}$ & $(\beta_{s,f,g}=\frac{1}{s},0)$ & $c_f=c_g=c_h=\alpha_h=1, \alpha_f = 0, \alpha_g=0$ \\[5mm]
                                                                  & $G_{SC}$ & $(\infty_h,\gamma_{s,f,g,h}=1-s)$ & $c_f=c_g=c_h=\alpha_h=1, \alpha_f = 0, \alpha_g=1$ \\[5mm]
                                                                  & $H_{SC}$ & $(\infty_h,\gamma_h=1)$ & $c_f=c_g=c_h=\alpha_h=1, \alpha_f = 0, \alpha_g>1$ \\[5mm]
                                                                  & $J_{SC}$ & $(\infty_{f,g},0)$ & $c_f=c_g=c_h=\alpha_h=1, \alpha_f = 0, 0 < \alpha_g < 1$ \\[5mm]\hline
 \textit{Scaled speedup model} \citep{schmidt2017parallel} & $A_{SC}$ & $(1,0)$ & $c_h=\alpha_h=1, \alpha_g-\alpha_f < 0$ \\[5mm]
                                                           & $B_{SC}$ & $(\beta_{s,f,g},0)$ & $c_h=\alpha_h=1, \alpha_g-\alpha_f = 0$ \\[5mm]
                                                           & $G_{SC}$ & $(\infty_h,\gamma_{s,f,g,h})$ & $c_h=\alpha_h=1, \alpha_g-\alpha_f = 1$ \\[5mm]
                                                           & $H_{SC}$ & $(\infty_h,\gamma_h=1)$ & $c_h=\alpha_h=1, \alpha_g-\alpha_f > 1$ \\[5mm]
                                                           & $J_{SC}$ & $(\infty_{f,g},0)$ & $c_h=\alpha_h=1, 0 < \alpha_g-\alpha_f < 1$ \\[5mm]

%

 \hline
\end{tabular}
\end{spacing}
\end{table}

A key issue for researchers is the assignment of their particular application to a scalability type, which requires determining the sequential workload $s$ and the power functions $f,g,h$ (see Definition \ref{def:powerFunctions}). In order to determine $s$, a straightforward approach is to execute the application on a single PU and measure the execution times $t_{seq}$ ans $t_{par}$ of the sequential and parallelizable workloads, resp., leading to $s=\frac{t_{seq}}{t_{seq}+t_{par}}$ and $p=1-s$.
\color{red}
I used this approach in the two computational experiments. As the second experiment shows, this approach can lead to inaccurate  results when the total execution time is small and numerical problems occur. As Table \ref{tab:scalabilityCases} shows, the value of $s$ affects speedup and efficiency bounds in some scalability cases. Under such circumstances, it may be helpful to use ``safe'' lower and upper bounds on $s$, and to define intervals of speedup and efficiency bounds.             

The determination of the power functions $f$, $g$ and $h$ can be much more challenging, depending on the algorithm used. While in the first experiment the functions could be determined straightforward, the second experiment shows that analytical, numerical and statistical approaches may be required to obtain reliable estimates of the functions. The application of such approaches may become quite tedious or even impossible due to the complexity of the parallelized algorithm. In this case, researchers are advised to use computational experiments to determine the functions. These issues limit the practical applicability and usefulness of the proposed model. 

 If the application involves data processing and analysis, the amount of data to be processed, and thus the parallelizable workload given by $p \cdot g(N)$, should be relatively easy to determine. Also, it seems reasonable to expect that the parallelizable workload can be almost equally distributed across the available PPUs ($h(N)\approx N$), unless the data processing requires taking into account data dependencies. In contrast, in numerical algorithms, which can be as simple as LU matrix decomposition as used in the second experiment, the determination of the functions $g$ and $h$ can become quite complicated. Also, in an optimization context, such as solving an instance of a mixed-integer linear program to optimality with a branch-and-bound algorithm, both the parallelizable workload $p \cdot g(N)$ and the effect of parallelization expressed by $h(N)$ may depend not only on $N$, which could be related to the size of the (optimization) problem instance to be solved, but also on the instance itself. For example, while some instances of a given problem may show (sub)linear speedup, other instances of the same size may benefit from superlinear speedup (see, for example, \citep{Rauchecker_Schryen_2019}).   

\subsection{Consideration of parallelization overhead}
\label{sec:consParallOverhead}

In the speedup and efficiency equations (\ref{eq:genericSpeedupEquation})-(\ref{eq:genericEfficiencyEquation}) and the resulting analysis, any overhead due to parallelization has been omitted for a variety of reasons. However, as can be seen from the results of the second computational experiment and widely acknowledged in the literature, parallelization overhead can cause large scalability degradation and have a significantly large impact on speedup and efficiency. Thus, the bounds may become loose. 

Overhead can be caused by several phenomena, including the existence of critical regions (exclusive access for only one process), inter-process communication, the creation and management of threads, and sequential-to-parallel synchronization due to data exchange \citep{yavits2014effect}. Such phenomena can be analyzed by considering an overhead function in the determination of parallel execution times, which are likely to depend, among other factors, on the number of parallel PUs.

In the literature, several ways of incorporating overhead functions into execution time evaluation have been proposed. One option is to include an additive overhead term in the speedup and efficiency functions (e.g., \citep{flatt1989performance,huang2013extending,pei2016extending}); an alternative approach is to use a multiplicative term (coefficient function) to account for the increased workload of parallel execution due to parallelization overhead (e.g., \citep{eyerman2010modeling,sun1990another,huang2013extending}).
Although I focus here on additive overhead functions, the key concepts, opportunities, and challenges for considering parallelization overhead also apply to multiplicative functions.

Our general speedup and efficiency equations ( \ref{eq:generalSpeedupEquation})-( \ref{eq:generalEfficiencyEquation}) already account for additive overhead with the term $z(N)$. Assuming that parallelization overhead can also be determined by a polynomial function $z(N):=c_z \cdot N^{\alpha_z}\;(c_z>0,\alpha_z \geq 0)$, my generic speedup and efficiency equations (\ref{eq:genericSpeedupEquation})-(\ref{eq:genericEfficiencyEquation}) would have to be changed to
\begin{equation}
\label{eq:genericSpeedupEquationWithOverhead}
S(N)=\frac{\pol{s \cdot c_f}{\alpha_f} + \pol{p \cdot c_g}{\alpha_g}}{\pol{s \cdot c_f}{\alpha_f} + \pol{\frac{p \cdot c_g}{c_h}}{(\alpha_g - \alpha_h)}+c_z \cdot N^{\alpha_z}}
\end{equation}
and
\begin{equation}
\label{eq:genericEfficiencyEquationWithOverhead}
E(N)=\frac{\pol{s \cdot c_f}{\alpha_f} + \pol{p \cdot c_g}{\alpha_g}}{\pol{s \cdot c_f}{(\alpha_f + 1)} + \pol{\frac{p \cdot c_g}{c_h}}{(\alpha_g - \alpha_h + 1)}+\cdot (c_z \cdot N^{\alpha_z+1})},
\end{equation}
respectively.

In the study of \citet{flatt1989performance}, the authors suggest that an additive (and continuous) overhead function should satisfy some mathematical assumptions. Under these assumptions, some theoretical results can be derived. One important result is that when Amdahl's setting is extended to include overhead, the speedup has a unique maximum at $N \geq 1$. In Appendix \ref{appendixSec:C}, I list the assumptions and prove that any polynomial overhead function $z(N)=c_z\cdot N^{\alpha_z}-c_z\;(c_z,\alpha_z >0, N \geq 1)$ satisfies all assumptions, so that the above result holds for modified equation (\ref{eq:genericSpeedupEquationWithOverhead}) in the Amdahl setting, if we exclude the case $\alpha_z = 0$ and allow the expansion of $z(N)$ by subtracting the constant $c_z$. \citet{flatt1989performance} also gives theoretical bounds on the (scaled) speedup with increasing problem size. 

\citet{huang2013extending} also suggest using an additive overhead function under the Amdahl setting that accounts for the data transmission overhead in multi-core environments. They propose using an overhead function $z(i)=f_t + \frac{g_t}{i}$, where $i$ is the number of communication links of a single core, and $f_t$ and $g_t$ are the sequential and parallel parts of the transmission, respectively. Applying their extended Amdahl model to a setting with area constraints \citep{hill2008amdahl}, they look for zero points of the speedup derivative to identify the optimal speedup.

A key conclusion from these findings is that if an additive overhead function is considered in parallel execution time and performance analysis, then functions for execution time, speedup, and efficiency may become non-monotonous, and as a consequence, the determination of limits and bounds with asymptotic analysis needs to be replaced or complemented by an analysis of extreme points; i.e., an analysis that accounts for parallelization overhead should determine the optimal number of PPUs for a given metric such as parallel execution time, speedup, and efficiency. This approach would not only lead to better predictions, but would also immediately lead to a suggestion of the appropriate number of parallel PUs to choose. Using my general speedup and efficiency equations (\ref{eq:generalSpeedupEquation})-( \ref{eq:generalEfficiencyEquation}), this would lead to solving the optimization problems (\ref{eq:generalExecutionTimeOptimization})-(\ref{eq:generalEfficiencyOptimization}) for execution time, speedup, and efficiency, respectively.
\begin{equation}
\label{eq:generalExecutionTimeOptimization}
\min_{N\in \mathbb{N}} \left( s\cdot f(N)+\frac{p\cdot g(N)}{h(N)}+z(N) \right)
\end{equation}
\begin{equation}
\label{eq:generalSpeedupOptimization}
\max_{N\in \mathbb{N}} S(N)= \max_{N\in \mathbb{N}} \left( \frac{s\cdot f(1)+p\cdot g(1)}{s\cdot f(N)+\frac{p\cdot g(N)}{h(N)}+z(N)} \right)
\end{equation}
\begin{equation}
\label{eq:generalEfficiencyOptimization}
\max_{N\in \mathbb{N}} E(N)= \max_{N\in \mathbb{N}} \left( \frac{s\cdot f(1)+p\cdot g(1)}{N \cdot (s\cdot f(N)+\frac{p\cdot g(N)}{h(N)}+z(N))} \right)
\end{equation}  
The determination of an additive overhead function $z$ should take into account the (architecture of the) parallel system, including the speed of the cores, the size and structure of the caches, and the operating system \citep{Brown2000}. Also, it should consider the roots of the parallelization overhead. For example, \citet{flatt1989performance} suggest using different overhead functions for i) scheduling in shared memory multiprocessors using critical regions, ii) synchronization in an array of processors arranged as a k-cube, and iii) synchronization in an array of processors connected by a logarithmic network. Additive terms were also used to account for the overhead associated with data preparation, communication, and synchronization \citep{li1988analysis,al2020amdahl,pei2016extending,yavits2014effect,huang2013extending}.

In summary, a performance analysis that takes into account parallelization overhead needs to consider
the root(s) of the parallelization overhead and the parallel system architecture applicable in a given context to determine the mathematical structure of the overhead function, its mathematical embedding in the computation of execution times (additive, multiplicative, etc.). From a methodological perspective, it requires searching for extreme points due to possible non-monotonicity of execution times, speedup and efficiency functions. 

\color{black}
\section{Conclusion}
\label{sec:conclusion}

In this work, I provide a generic speedup (and thus also efficiency) model, which generalizes many prominent models suggested in the literature and allows showing that they can be considered special cases with different assumptions of a unifying approach. The genericity of the speedup model is achieved through parameterization. Considering combinations of parameter ranges, I identify six different asymptotic speedup cases and eight different asymptotic efficiency cases; these cases include sublinear, linear and superlinear speedup and efficiency. Based upon the identified speedup and efficiency cases, I derive eleven different scalability cases and types, to which instantiations of my generic speedup (and efficiency) model may lead. Researchers can draw upon my suggested typology to classify their speedup model and/or to determine the asymptotic scalability of their application when the number of PPUs increases. \color{red} Also, the description of two computational experiments demonstrates the practical application of the model and the typology. \color{black}

My theoretical analysis is based upon several assumptions which are common in the literature (e.g., \citep{sun1993scalable}). First, I assume that the overall workload only contains two parts, a sequential part and a perfectly parallelizable part, which can be executed in parallel on all available PPUs. In practice, the latter condition may not always hold but even then my speedup and efficiency results are useful as they can be used as upper bounds of achievable speedup and efficiency. Alternative models that do not require the above dichotomy distinction have been proposed in the literature, including parallelism/span-and-work models, multiple-fraction models, and roofline models (e.g., \citep[p. 772ff]{cormen2022introduction}, \citep[p. 141ff]{al2020amdahl}), \citep{cassidy2011beyond}). Future theoretical analysis of speedup and efficiency limits may consider those types of models. 


Second, while my generic speedup model includes a function for parallelization overhead, I assume that this overhead is negligible and omit this function from my analysis. However, I admit that parallelization overhead may cause large scalability degradation \citep{huang2013extending} and have considerably large effects on speedup and efficiency.
\color{red}
A discussion of the implications and research avenues for accounting for parallelization overhead is provided in Section \ref{sec:consParallOverhead}. Future research can build on these ideas to obtain tighter bounds on speedup and efficiency, and to derive recommendations for the optimal choice of the number of PPUs to use.  

\color{black}

Third, in my analysis I focus on homogeneous parallel computing environments. I acknowledge that, in modern parallel computing environments, parallel processing units are not necessarily equally potential in their computing capabilities and that a substantial body of literature on speedup in such heterogeneous computing environments exist; see, for example, the surveys on heterogeneous multicore environments of \citet{al2013survey} and \citet{al2020amdahl}. Many works suggest extensions of \textit{Amdahl's law}, \textit{Gustafson's law} and/or \textit{Sun and Ni’s Law} for such environments (e.g., \citep{hill2008amdahl,juurlink2012amdahl,Ye2013,rafiev2018speedup,zidenberg2013optimal,morad2012generalized,moncrieff1996heterogeneous}). Studies on speedup and efficiency properties of architecture-dependent laws are particularly helpful for the design of multi-core environments. Future work may extend my generic speedup model by concepts of different types of PPUs as suggested in the literature and adapt my theoretical analysis to heterogeneous settings.

Finally, merging the two abovementioned research streams leads to the consideration of speedup and efficiency in heterogeneous parallel computing environments under the consideration of overhead (functions). My model and theoretical analysis may be extended in both regards, drawing on prior work. For example, \citet{huang2013extending} suggest an extension of \textit{Amdahl's law} and \textit{Gustafson's law} in architecture-specific multi-core settings by considering communication overhead and area constraints; \citet{pei2016extending} extend \textit{Amdahl's law} for heterogeneous multicore processors with the consideration of overhead through data preparation; and \citet{morad2006performance} analyze overhead as a result of synchronization, communication and coherence costs based upon Amdahl's model for asymmetric cluster chip multiprocessors.

With the suggestion of a generic and unifying speedup (and efficiency) model and its asymptotic analysis, I hope to provide a theoretical basis for and typology of scalability of parallel algorithms in homogeneous computing environments. Future research can draw upon and extend my research results to address various extensions of my setting.	

     
\bibliographystyle{abbrvnat}

\bibliography{bib/performanceParallelization}

\newpage
\appendix

\section{Calculations of speedup limits}
\label{appendixSec:A}

We rewrite (I) as follows:
\begin{equation}
\label{eq:I_Def}
(I)=\frac{\pol{s \cdot c_f}{\alpha_f} }{\pol{s \cdot c_f}{\alpha_f} + \pol{\frac{p \cdot c_g}{c_h}}{(\alpha_g - \alpha_h)}}=
\frac{s \cdot c_f}{s \cdot c_f + \pol{\frac{p \cdot c_g}{c_h}}{\alpha_g - \alpha_h - \alpha_f}}
\end{equation}

and obtain 
\begin{empheq}[left={\lim_{N\rightarrow \infty} (I)=}\empheqlbrace]{align}
1 &,\; \alpha_f > \alpha_g - \alpha_h \label{eq:IDef_a_Final}\\
\frac{s \cdot c_f}{s \cdot c_f + \frac{p \cdot c_g}{c_h}} & , \alpha_f = \alpha_g - \alpha_h \label{eq:IDef_b_Final}\\
0  & , \alpha_f < \alpha_g - \alpha_h \label{eq:IDef_c_Final}
\end{empheq}

It should be noticed that, in contrast to the limits in equations (\ref{eq:IDef_a_Final}) and (\ref{eq:IDef_b_Final}), which show upper bounds, the limit in equation (\ref{eq:IDef_c_Final}) represents a lower bound.

We rewrite (II) as follows:
\begin{empheq}[left={(II)=\frac{\pol{p \cdot c_g}{\alpha_g} }{\pol{s \cdot c_f}{\alpha_f} + \pol{\frac{p \cdot c_g}{c_h}}{(\alpha_g - \alpha_h)}}=}\empheqlbrace]{align}
\frac{p \cdot c_g}{s \cdot \pol{c_f}{(\alpha_f - \alpha_g)} + \frac{p \cdot c_g}{\pol{c_h}{\alpha_h}}}
&,\; \alpha_f > \alpha_g \label{eq:II_Def_a}\\
\frac{p \cdot c_g}{s \cdot c_f + \frac{p \cdot c_g}{\pol{c_h}{\alpha_h}}}
&,\; \alpha_f = \alpha_g \label{eq:II_Def_b}\\
\frac{p \cdot \pol{c_g}{(\alpha_g - \alpha_f)}}{s \cdot c_f + \frac{p \cdot c_g}{c_h} \cdot \pol{c_h}{(\alpha_g - \alpha_h - \alpha_f)}}
&,\; \alpha_f < \alpha_g \label{eq:II_Def_c}
\end{empheq}

For equation (\ref{eq:II_Def_a}), we obtain 
\begin{equation}
\label{eq:II_Def_a_Final}
\lim_{N\rightarrow \infty} (II)= 0,\; \alpha_f > \alpha_g
\end{equation}

For equation (\ref{eq:II_Def_b}), we obtain 
\begin{empheq}[left={\lim_{N\rightarrow \infty} (II)=}\empheqlbrace]{align}
\frac{p \cdot c_g}{s \cdot c_f}
&,\; \alpha_f = \alpha_g, \alpha_h > 0 \label{eq:II_Def_b1_Final}\\
\frac{p \cdot c_g}{s \cdot c_f + \frac{p \cdot c_g}{c_h}}
&,\; \alpha_f = \alpha_g, \alpha_h = 0 \label{eq:II_Def_b2_Final}
\end{empheq}

For equation (\ref{eq:II_Def_c}), we obtain 
\small
\begin{empheq}[left={\lim_{N\rightarrow \infty} (II)=}\empheqlbrace]{align}
\lim_{N\rightarrow \infty} \frac{p \cdot \pol{c_g}{\alpha_h}}{\frac{s \cdot c_f}{\pol{}{(\alpha_g - \alpha_h - \alpha_f)}} + \frac{p \cdot c_g}{c_h}} = \infty\;(\Theta (N^{\alpha_h}))
&,\; \alpha_f < \alpha_g, \alpha_h > 0, \alpha_g - \alpha_h - \alpha_f \geq 0 \label{eq:II_Def_b3a_Final}\\
\lim_{N\rightarrow \infty} \frac{p \cdot \pol{c_g}{(\alpha_g - \alpha_f)}}{s \cdot c_f + \frac{p \cdot c_g}{c_h} \cdot \pol{}{(\alpha_g - \alpha_h - \alpha_f)}} = \infty\;(\Theta (N^{(\alpha_g - \alpha_f)}))
&,\; \alpha_f < \alpha_g, \alpha_h > 0, \alpha_g - \alpha_h - \alpha_f < 0 \label{eq:II_Def_b3b_Final}\\
\lim_{N\rightarrow \infty} \frac{p \cdot c_g}{\frac{s \cdot c_f}{\pol{}{(\alpha_g - \alpha_f)}} + \frac{p \cdot c_g}{c_h}} = c_h
&,\; \alpha_f < \alpha_g, \alpha_h = 0 \label{eq:II_Def_b3c_Final}\end{empheq}

\normalsize

It should be noticed that, in contrast to the limits in equations (\ref{eq:II_Def_b1_Final}), (\ref{eq:II_Def_b2_Final}) and (\ref{eq:II_Def_b3c_Final}), which show upper bounds, the limit in equation (\ref{eq:II_Def_a_Final}) represents a lower bound.

\section{Calculations of efficiency limits}
\label{appendixSec:B}

\begin{empheq}[left={\lim_{N\rightarrow \infty} (I')=\lim_{N\rightarrow \infty} \frac{s \cdot c_f}{s \cdot c_f \cdot N + \pol{\frac{p \cdot c_g}{c_h}}{(\alpha_g - \alpha_h - \alpha_f + 1)}}=}\empheqlbrace]{align}
0&,\; \alpha_g - \alpha_h - \alpha_f + 1 > 0\label{eq:I_eff_1}\\
0&,\; \alpha_g - \alpha_h - \alpha_f + 1 = 0\label{eq:I_eff_2}\\
0&,\; \alpha_g - \alpha_h - \alpha_f + 1 < 0\label{eq:I_eff_3}
\end{empheq}

It should be noticed that the limits in equations (\ref{eq:I_eff_1}) to (\ref{eq:I_eff_3}) all show lower bounds.

We rewrite $(II')$ as follows:
\small
\begin{empheq}[left={(II')=\frac{\pol{p \cdot c_g}{\alpha_g}}{\pol{s \cdot c_f}{(\alpha_f + 1)} + \pol{\frac{p \cdot c_g}{c_h}}{(\alpha_g - \alpha_h + 1)}}=}\empheqlbrace]{align}
\frac{p \cdot c_g \cdot \pol{}{(\alpha_g - \alpha_f - 1)}}{s \cdot c_f + \frac{p \cdot c_g}{c_h}\cdot \pol{}{(\alpha_g - \alpha_h - \alpha_f)}}
&,\; \alpha_f < \alpha_g - 1, \alpha_h > 1, \alpha_f\geq \alpha_g - \alpha_h \label{eq:II_eff_Def_a}\\
\frac{p \cdot c_g \cdot \pol{}{(\alpha_h - 1)}}{\frac{s \cdot c_f}{\pol{}{(\alpha_g - \alpha_f -1)}} + \frac{p \cdot c_g}{c_h}}
&,\; \alpha_f < \alpha_g - 1, \alpha_h > 1, \alpha_f <\alpha_g - \alpha_h  \label{eq:II_eff_Def_b}\\
\frac{p \cdot c_g}{\pol{s \cdot c_f}{(\alpha_f - \alpha_g + 1)} + \pol{\frac{p \cdot c_g}{c_h}}{(- \alpha_h + 1)}} &,\; \text{else} \label{eq:II_eff_Def_c}
\end{empheq}

\normalsize

and obtain
\begin{empheq}[left={\lim_{N\rightarrow \infty} (II')=}\empheqlbrace]{align}
0 &,\; \alpha_f > \alpha_g - 1, 0 \leq \alpha_h \leq 1 \label{eq:II_eff_1}\\
0 &,\; \alpha_f > \alpha_g - 1, \alpha_h > 1 \label{eq:II_eff_2}\\
0 &,\; \alpha_f = \alpha_g - 1, 0 \leq \alpha_h \leq 1 \label{eq:II_eff_3}\\
\frac{p \cdot c_g}{s \cdot c_f + \frac{p \cdot c_g}{c_h}} &,\; \alpha_f = \alpha_g - 1, \alpha_h = 1 \label{eq:II_eff_4}\\
\frac{p \cdot c_g}{s \cdot c_f} &,\; \alpha_f = \alpha_g - 1, \alpha_h > 1 \label{eq:II_eff_5}\\
0 &,\; \alpha_f < \alpha_g - 1, 0 \leq \alpha_h < 1 \label{eq:II_eff_6}\\
c_h &,\; \alpha_f < \alpha_g - 1, \alpha_h = 1 \label{eq:II_eff_7}\\
\infty \;(\Theta (N^{\alpha_g - \alpha_f -1})) & ,\; \alpha_f < \alpha_g - 1, \alpha_h > 1, \alpha_f > \alpha_g - \alpha_h \label{eq:II_eff_8}\\
\infty \;(\Theta (N^{\alpha_h -1})) & ,\; \alpha_f < \alpha_g - 1, \alpha_h > 1, \alpha_f = \alpha_g - \alpha_h \label{eq:II_eff_9}\\
\infty \;(\Theta (N^{\alpha_h -1})) & ,\; \alpha_f < \alpha_g - 1, \alpha_h > 1, \alpha_f < \alpha_g - \alpha_h \label{eq:II_eff_10}
\end{empheq}

It should be noticed that the limits in equations (\ref{eq:II_eff_1}) to (\ref{eq:II_eff_6}) all show lower bounds while the limit in equation (\ref{eq:II_eff_7}) is an upper bound.

\color{red}

\section{Computational experiments}
\label{sec:AppendixExperiments}

\subsection{Experiment 1: Parallel matrix multiplication with fixed workload}
\label{sec:AppendixExperiment1}

\clearpage

\begin{table}[htb]
    \caption{Computational results of matrix multiplication with fixed size of all matrices}
    \label{tab:resultsMatrixMultFixedSize}
    \begin{spacing}{1.2}
    \begin{tabular}{|c|c|c|c|c|c|}
        \hline
        \multirow{2}{*}{N} & \multirow{2}{*}{T(N) [in ms]}  & \multicolumn{2}{c|}{S(N)} & \multicolumn{2}{c|}{E(N)} \\
        \cline{3-6}
        & & Theor. & Comp. & Theor. & Comp. \\
        \hline
        1   & 1,529,020 & -- & -- & -- & -- \\
        \hline
        2   & 953,760 & 1.953898 & 1.603150 & 0.976949 & 0.801575 \\
        \hline
        4   & 493,262 & 3.735577 & 3.099813 & 0.933894 & 0.774953 \\
        \hline
        8   & 270,447 & 6.865980 & 5.653677 & 0.858248 & 0.706710 \\
        \hline
        16  & 163,341 & 11.817493 & 9.360908 & 0.738593 & 0.585057 \\
        \hline
        32  & 100,269 & 18.481672 & 15.249180 & 0.577552 & 0.476537 \\
        \hline
        64  & 74,392 & 25.739145 & 20.553555 & 0.402174 & 0.321149 \\
        \hline
        128 & 64,154 & 32.027504 & 23.833588 & 0.250215 & 0.186200 \\
        \hline
    \end{tabular}
    \end{spacing}
\end{table}

\subsection{Experiment 2: Parallel LU decomposition with variable workload}
\label{sec:AppendixExperiment2}

Using the algorithm shown in Figure \ref{fig:AlgLUDecomp}, I determine the number of calculations along the outer loop that iterates over the variable $i$. In the first iteration ($i=1$), the elements $L_{j,1}$ ($j=2 \ldots z$) and the elements $U_{j,l}$ ($j,l=2 \ldots z$) are calculated and assigned, for a total of $(z-1) \cdot z$ calculations. Similarly, in the second iteration ($i=2$), the elements $L_{j,2}$ ($j=3 \ldots z$) and the elements $U_{j,l}$ ($j,l=3 \ldots z$) are calculated and assigned; i.e., a total of $(z-2) \cdot (z-1)$ calculations are required. In general, in iteration $i$, the number of required calculations equals $(z-i) \cdot (z-i+1)$. This yields
\begin{alignat*}{3}\hat{g}(z) \quad && = && \quad & \sum_{i=1}^{z-1} \left[ (z-i) \cdot (z-i+1) \right]
= \sum_{i=1}^{z-1} \left[ z^2 - z\cdot i + z - z\cdot i +i^2 - i \right] \\
           && = &&       &\sum_{i=1}^{z-1} \left[ z^2 + z + i \cdot (-z - z - 1) + i^2 \right]
= \sum_{i=1}^{z-1} \left[ z^2 + z \right] + \sum_{i=1}^{z-1} \left[(-2z - 1) \cdot i \right] +  \sum_{i=1}^{z-1}  i^2 \\
           && = &&       &(z-1) \cdot (z^2 + z) - (2z + 1) \cdot \sum_{i=1}^{z-1} i + \sum_{i=1}^{z-1} i^2 \\
           && = &&       &z^3 + z^2 -z^2 - z - (2z + 1) \cdot \frac{z \cdot (z - 1)}{2} + \frac{z \cdot (z-1) \cdot (2z-1)}{6}\\
           && = &&       &z^3 - z + \frac{z \cdot (z-1) \cdot \left[ (2z-1) \cdot 3 \cdot (2z+1) \right]}{6}\\
           && = &&       &z^3 - z + \frac{(z^2-z) \cdot (-4z-4)}{6}
= z^3 - z + \frac{(z^2-z) \cdot (-2z-2)}{3}\\
           && = &&       &z^3 - z + \frac{(-2z^3+2z)}{3}
= z^3 - z - \frac{2}{3} z^3 + \frac{2}{3} z = \frac{1}{3} z^3 + \frac{1}{3} z = \frac{1}{3} \left( z^3-z \right)                   
\end{alignat*}

In the context of determining the scaling function $h$, the Tables \ref{tab:dataPointsLRDecomp100}-\ref{tab:dataPointsLRDecomp100000} show data points of $N$, $\hat{g}(N,z_1)$, $\hat{g}_{reduced}(N,z_1)$, $\hat{h}(N,z_1)$, and $(\hat{h}(N,z_1)/N)$ for various values of the number of rows/columns ($z_1$) of the input matrix.  

\begin{table}[htb]
\caption{Data points of LR decomposition $(z_1=$100)}
\label{tab:dataPointsLRDecomp100}
\begin{tabular}{|c|c|c|c|c|c|c|}
\hline
$i$ & $N=2^i$ & $z$ & $\hat{g}(N,z_1)$ & $\hat{g}_{reduced}(N,z_1)$ & $\hat{h}(N,z_1)$ & $\hat{h}(N,z_1) / N$ \\
\hline
$1$ & $2$ & $200$ & $2.6666e+06$ & $1.33835e+06$ & $1.99245$ & $0.996227$ \\
\hline
$2$ & $4$ & $400$ & $2.13332e+07$ & $5.3634e+06$ & $3.97755$ & $0.994388$ \\
\hline
$3$ & $8$ & $800$ & $1.70666e+08$ & $2.14733e+07$ & $7.94784$ & $0.99348$ \\
\hline
$4$ & $16$ & $1,600$ & $1.36533e+09$ & $8.59323e+07$ & $15.8885$ & $0.993029$ \\
\hline
$5$ & $32$ & $3,200$ & $1.09227e+10$ & $3.43807e+08$ & $31.7697$ & $0.992805$ \\
\hline
$6$ & $64$ & $6,400$ & $8.73813e+10$ & $1.37538e+09$ & $63.5323$ & $0.992693$ \\
\hline
$7$ & $128$ & $12,800$ & $6.99051e+11$ & $5.50185e+09$ & $127.057$ & $0.992636$ \\
\hline
$8$ & $256$ & $25,600$ & $5.59241e+12$ & $2.2008e+10$ & $254.108$ & $0.992608$ \\
\hline
$9$ & $512$ & $51,200$ & $4.47392e+13$ & $8.80333e+10$ & $508.208$ & $0.992594$ \\
\hline
$10$ & $1024$ & $102,400$ & $3.57914e+14$ & $3.52136e+11$ & $1016.41$ & $0.992587$ \\
\hline
$11$ & $2048$ & $204,800$ & $2.86331e+15$ & $1.40855e+12$ & $2032.81$ & $0.992584$ \\
\hline
$12$ & $4096$ & $409,600$ & $2.29065e+16$ & $5.6342e+12$ & $4065.62$ & $0.992582$ \\
\hline
$13$ & $8192$ & $819,200$ & $1.83252e+17$ & $2.25368e+13$ & $8131.23$ & $0.992581$ \\
\hline
$14$ & $16384$ & $1.6384e+06$ & $1.46602e+18$ & $9.01473e+13$ & $16262.4$ & $0.992581$ \\
\hline
$15$ & $32768$ & $3.2768e+06$ & $1.17281e+19$ & $3.60589e+14$ & $32524.9$ & $0.992581$ \\
\hline
$16$ & $65536$ & $6.5536e+06$ & $9.3825e+19$ & $1.44236e+15$ & $65049.8$ & $0.992581$ \\
\hline
$17$ & $131072$ & $1.31072e+07$ & $7.506e+20$ & $5.76943e+15$ & $130100$ & $0.992581$ \\
\hline
$18$ & $262144$ & $2.62144e+07$ & $6.0048e+21$ & $2.30777e+16$ & $260199$ & $0.99258$ \\
\hline
$19$ & $524288$ & $5.24288e+07$ & $4.80384e+22$ & $9.23109e+16$ & $520398$ & $0.99258$ \\
\hline
$20$ & $1.04858e+06$ & $1.04858e+08$ & $3.84307e+23$ & $3.69243e+17$ & $1.0408e+06$ & $0.99258$ \\
\hline
\end{tabular}
\end{table}

\begin{table}
\caption{Data points of LR decomposition $(z_1=$1,000)}
\label{tab:dataPointsLRDecomp1000}
\begin{tabular}{|c|c|c|c|c|c|c|}
\hline
$i$ & $N=2^i$ & $z$ & $\hat{g}(N,z_1)$ & $\hat{g}_{reduced}(N,z_1)$ & $\hat{h}(N,z_1)$ & $\hat{h}(N,z_1) / N$ \\
\hline
$1$ & $2$ & $2,000$ & $2.66667e+09$ & $1.33383e+09$ & $1.99925$ & $0.999625$ \\
\hline
$2$ & $4$ & $4,000$ & $2.13333e+10$ & $5.33633e+09$ & $3.99775$ & $0.999438$ \\
\hline
$3$ & $8$ & $8,000$ & $1.70667e+11$ & $2.13473e+10$ & $7.99475$ & $0.999344$ \\
\hline
$4$ & $16$ & $16,000$ & $1.36533e+12$ & $8.53933e+10$ & $15.9888$ & $0.999297$ \\
\hline
$5$ & $32$ & $32,000$ & $1.09227e+13$ & $3.41581e+11$ & $31.9768$ & $0.999274$ \\
\hline
$6$ & $64$ & $64,000$ & $8.73813e+13$ & $1.36634e+12$ & $63.9528$ & $0.999262$ \\
\hline
$7$ & $128$ & $128,000$ & $6.99051e+14$ & $5.4654e+12$ & $127.905$ & $0.999257$ \\
\hline
$8$ & $256$ & $256,000$ & $5.59241e+15$ & $2.18616e+13$ & $255.809$ & $0.999254$ \\
\hline
$9$ & $512$ & $512,000$ & $4.47392e+16$ & $8.74467e+13$ & $511.617$ & $0.999252$ \\
\hline
$10$ & $1024$ & $1.024e+06$ & $3.57914e+17$ & $3.49787e+14$ & $1023.23$ & $0.999252$ \\
\hline
$11$ & $2048$ & $2.048e+06$ & $2.86331e+18$ & $1.39915e+15$ & $2046.47$ & $0.999251$ \\
\hline
$12$ & $4096$ & $4.096e+06$ & $2.29065e+19$ & $5.5966e+15$ & $4092.93$ & $0.999251$ \\
\hline
$13$ & $8192$ & $8.192e+06$ & $1.83252e+20$ & $2.23864e+16$ & $8185.86$ & $0.999251$ \\
\hline
$14$ & $16384$ & $1.6384e+07$ & $1.46602e+21$ & $8.95456e+16$ & $16371.7$ & $0.999251$ \\
\hline
$15$ & $32768$ & $3.2768e+07$ & $1.17281e+22$ & $3.58182e+17$ & $32743.5$ & $0.999251$ \\
\hline
$16$ & $65536$ & $6.5536e+07$ & $9.3825e+22$ & $1.43273e+18$ & $65486.9$ & $0.999251$ \\
\hline
$17$ & $131072$ & $1.31072e+08$ & $7.506e+23$ & $5.73092e+18$ & $130974$ & $0.999251$ \\
\hline
$18$ & $262144$ & $2.62144e+08$ & $6.0048e+24$ & $2.29237e+19$ & $261948$ & $0.999251$ \\
\hline
$19$ & $524288$ & $5.24288e+08$ & $4.80384e+25$ & $9.16947e+19$ & $523895$ & $0.999251$ \\
\hline
$20$ & $1.04858e+06$ & $1.04858e+09$ & $3.84307e+26$ & $3.66779e+20$ & $1.04779e+06$ & $0.999251$ \\
\hline
\end{tabular}
\end{table}

\begin{table}
\caption{Data points of LR decomposition $(z_1=$10,000)}
\label{tab:dataPointsLRDecomp10000}
\begin{tabular}{|c|c|c|c|c|c|c|}
\hline
$i$ & $N=2^i$ & $z$ & $\hat{g}(N,z_1)$ & $\hat{g}_{reduced}(N,z_1)$ & $\hat{h}(N,z_1)$ & $\hat{h}(N,z_1) / N$ \\
\hline
$1$ & $2$ & $20,000$ & $2.66667e+12$ & $1.33338e+12$ & $1.99992$ & $0.999962$ \\
\hline
$2$ & $4$ & $40,000$ & $2.13333e+13$ & $5.33363e+12$ & $3.99978$ & $0.999944$ \\
\hline
$3$ & $8$ & $80,000$ & $1.70667e+14$ & $2.13347e+13$ & $7.99948$ & $0.999934$ \\
\hline
$4$ & $16$ & $160,000$ & $1.36533e+15$ & $8.53393e+13$ & $15.9989$ & $0.99993$ \\
\hline
$5$ & $32$ & $320,000$ & $1.09227e+16$ & $3.41358e+14$ & $31.9977$ & $0.999927$ \\
\hline
$6$ & $64$ & $640,000$ & $8.73813e+16$ & $1.36543e+15$ & $63.9953$ & $0.999926$ \\
\hline
$7$ & $128$ & $1.28e+06$ & $6.99051e+17$ & $5.46174e+15$ & $127.99$ & $0.999926$ \\
\hline
$8$ & $256$ & $2.56e+06$ & $5.59241e+18$ & $2.1847e+16$ & $255.981$ & $0.999925$ \\
\hline
$9$ & $512$ & $5.12e+06$ & $4.47392e+19$ & $8.73879e+16$ & $511.962$ & $0.999925$ \\
\hline
$10$ & $1024$ & $1.024e+07$ & $3.57914e+20$ & $3.49552e+17$ & $1023.92$ & $0.999925$ \\
\hline
$11$ & $2048$ & $2.048e+07$ & $2.86331e+21$ & $1.39821e+18$ & $2047.85$ & $0.999925$ \\
\hline
$12$ & $4096$ & $4.096e+07$ & $2.29065e+22$ & $5.59282e+18$ & $4095.69$ & $0.999925$ \\
\hline
$13$ & $8192$ & $8.192e+07$ & $1.83252e+23$ & $2.23713e+19$ & $8191.39$ & $0.999925$ \\
\hline
$14$ & $16384$ & $1.6384e+08$ & $1.46602e+24$ & $8.94852e+19$ & $16382.8$ & $0.999925$ \\
\hline
$15$ & $32768$ & $3.2768e+08$ & $1.17281e+25$ & $3.57941e+20$ & $32765.5$ & $0.999925$ \\
\hline
$16$ & $65536$ & $6.5536e+08$ & $9.3825e+25$ & $1.43176e+21$ & $65531.1$ & $0.999925$ \\
\hline
$17$ & $131072$ & $1.31072e+09$ & $7.506e+26$ & $5.72705e+21$ & $131062$ & $0.999925$ \\
\hline
$18$ & $262144$ & $2.62144e+09$ & $6.0048e+27$ & $2.29082e+22$ & $262124$ & $0.999925$ \\
\hline
$19$ & $524288$ & $5.24288e+09$ & $4.80384e+28$ & $9.16328e+22$ & $524249$ & $0.999925$ \\
\hline
$20$ & $1.04858e+06$ & $1.04858e+10$ & $3.84307e+29$ & $3.66531e+23$ & $1.0485e+06$ & $0.999925$ \\
\hline
\end{tabular}
\end{table}

\begin{table}
\caption{Data points of LR decomposition $(z_1=$100,000)}
\label{tab:dataPointsLRDecomp100000}
\begin{tabular}{|c|c|c|c|c|c|c|}
\hline
$i$ & $N=2^i$ & $z$ & $\hat{g}(N,z_1)$ & $\hat{g}_{reduced}(N,z_1)$ & $\hat{h}(N,z_1)$ & $\hat{h}(N,z_1) / N$ \\
\hline
$1$ & $2$ & $200,000$ & $2.66667e+15$ & $1.33334e+15$ & $1.99999$ & $0.999996$ \\
\hline
$2$ & $4$ & $400,000$ & $2.13333e+16$ & $5.33336e+15$ & $3.99998$ & $0.999994$ \\
\hline
$3$ & $8$ & $800,000$ & $1.70667e+17$ & $2.13335e+16$ & $7.99995$ & $0.999993$ \\
\hline
$4$ & $16$ & $1.6e+06$ & $1.36533e+18$ & $8.53339e+16$ & $15.9999$ & $0.999993$ \\
\hline
$5$ & $32$ & $3.2e+06$ & $1.09227e+19$ & $3.41336e+17$ & $31.9998$ & $0.999993$ \\
\hline
$6$ & $64$ & $6.4e+06$ & $8.73813e+19$ & $1.36534e+18$ & $63.9995$ & $0.999993$ \\
\hline
$7$ & $128$ & $1.28e+07$ & $6.99051e+20$ & $5.46137e+18$ & $127.999$ & $0.999993$ \\
\hline
$8$ & $256$ & $2.56e+07$ & $5.59241e+21$ & $2.18455e+19$ & $255.998$ & $0.999993$ \\
\hline
$9$ & $512$ & $5.12e+07$ & $4.47392e+22$ & $8.7382e+19$ & $511.996$ & $0.999993$ \\
\hline
$10$ & $1024$ & $1.024e+08$ & $3.57914e+23$ & $3.49528e+20$ & $1023.99$ & $0.999993$ \\
\hline
$11$ & $2048$ & $2.048e+08$ & $2.86331e+24$ & $1.39811e+21$ & $2047.98$ & $0.999993$ \\
\hline
$12$ & $4096$ & $4.096e+08$ & $2.29065e+25$ & $5.59245e+21$ & $4095.97$ & $0.999993$ \\
\hline
$13$ & $8192$ & $8.192e+08$ & $1.83252e+26$ & $2.23698e+22$ & $8191.94$ & $0.999993$ \\
\hline
$14$ & $16384$ & $1.6384e+09$ & $1.46602e+27$ & $8.94792e+22$ & $16383.9$ & $0.999993$ \\
\hline
$15$ & $32768$ & $3.2768e+09$ & $1.17281e+28$ & $3.57917e+23$ & $32767.8$ & $0.999993$ \\
\hline
$16$ & $65536$ & $6.5536e+09$ & $9.3825e+28$ & $1.43167e+24$ & $65535.5$ & $0.999993$ \\
\hline
$17$ & $131072$ & $1.31072e+10$ & $7.506e+29$ & $5.72667e+24$ & $131071$ & $0.999993$ \\
\hline
$18$ & $262144$ & $2.62144e+10$ & $6.0048e+30$ & $2.29067e+25$ & $262142$ & $0.999992$ \\
\hline
$19$ & $524288$ & $5.24288e+10$ & $4.80384e+31$ & $9.16267e+25$ & $524284$ & $0.999993$ \\
\hline
$20$ & $1.04858e+06$ & $1.04858e+11$ & $3.84307e+32$ & $3.66507e+26$ & $1.04857e+06$ & $0.999993$ \\
\hline
\end{tabular}
\end{table}

\begin{table}[h]
    \centering
    \caption{Computational results of LU matrix decomposition with variable size of matrices}
    \label{tab:resultsMatrixLUDecompVarSize}
    \begin{spacing}{1.2}
    \begin{tabular}{|c|c|c|c|c|c|c|}
        \hline
        \multirow{2}{*}{N} & \multirow{2}{*}{\#(P)PUs} & \multirow{2}{*}{T(N) [in ms]} & \multicolumn{2}{c|}{S(N)} & \multicolumn{2}{c|}{E(N)} \\
        \cline{4-7}
        & & & Theor. & Comp. & Theor. & Comp. \\
        \hline
        1 & 1 & 2 & -- & -- & -- & -- \\
        \hline
        \multirow{2}{*}{2} & 1 & 21 & \multirow{2}{*}{1.997481} & \multirow{2}{*}{2.100000} & \multirow{2}{*}{0.998741} & \multirow{2}{*}{1.050000} \\
        \cline{2-3}
        & 2 & 10 & & & & \\
        \hline
        \multirow{2}{*}{4} & 1 & 167 & \multirow{2}{*}{3.998107} & \multirow{2}{*}{4.771429} & \multirow{2}{*}{0.999527} & \multirow{2}{*}{1.192857} \\
        \cline{2-3}
        & 4 & 35 & & & & \\
        \hline
        \multirow{2}{*}{8} & 1 & 1,053 & \multirow{2}{*}{7.998896} & \multirow{2}{*}{9.486486} & \multirow{2}{*}{0.999862} & \multirow{2}{*}{1.185811} \\
        \cline{2-3}
        & 8 & 111 & & & & \\
        \hline
        \multirow{2}{*}{16} & 1 & 10,255 & \multirow{2}{*}{15.999408} & \multirow{2}{*}{17.293423} & \multirow{2}{*}{0.999963} & \multirow{2}{*}{1.080839} \\
        \cline{2-3}
        & 16 & 593 & & & & \\
        \hline
        \multirow{2}{*}{32} & 1 & 94,539 & \multirow{2}{*}{31.999695} & \multirow{2}{*}{31.628973} & \multirow{2}{*}{0.999990} & \multirow{2}{*}{0.988405} \\
        \cline{2-3}
        & 32 & 2,989 & & & & \\
        \hline
        \multirow{2}{*}{64} & 1 & 831,699 & \multirow{2}{*}{63.999844} & \multirow{2}{*}{46.016323} & \multirow{2}{*}{0.999998} & \multirow{2}{*}{0.719005} \\
        \cline{2-3}
        & 64 & 18,074 & & & & \\
        \hline
        \multirow{2}{*}{128} & 1 & 5,383,229 & \multirow{2}{*}{127.999924} & \multirow{2}{*}{26.609865} & \multirow{2}{*}{0.999999} & \multirow{2}{*}{0.207890} \\
        \cline{2-3}
        & 128 & 202,302 & & & & \\
        \hline
    \end{tabular}
    \end{spacing}
    \#(P)PUs: number of (parallel) processing units.
\end{table}

\clearpage

\section{Mathematical requirements on overhead functions \citep{flatt1989performance}}
\label{appendixSec:C}

\citet{flatt1989performance} formulate the following five requirements on a parallelization overhead function $z: \mathcal{R} \rightarrow \mathcal{R}$:
\begin{enumerate}
    \item $z$ is continuous and twice differentiable with respect to $N$.
    \item $z(1)=0$
    \item $z'(N) > 0\quad \forall N \geq 1$
    \item $N\cdot z''(N)+2\cdot z'(N) > 0\quad \forall N \geq 1$
    \item There exists $N_1 \geq 1$ such that $z(N_1)=1$.
\end{enumerate}

 \begin{lemma}
 Each function $z(N):=c_z \cdot N^{\alpha_z}-c_z$ with $c_z,\alpha_z>0$ meets the above requirements for $N \geq 1$.
 \end{lemma}

 \begin{proof} I prove any of the five conditions separately:
  \begin{enumerate}
\item $z$ is apparently continuous and twice differentiable with $z'(N)=c_z \cdot \alpha_z \cdot N^{\alpha_z - 1}$ and $z''(N)=c_z \cdot \alpha_z \cdot (\alpha_z - 1) \cdot N^{\alpha_z - 2}\quad \forall \alpha_z,c_z>0$
    \item $z(1)=c_z \cdot 1^{\alpha_z}-c_z=0\quad \forall \alpha_z,c_z>0$
    \item $z'(N) = c_z \cdot \alpha_z \cdot N^{\alpha_z-1} >0\qquad \forall \alpha_z,c_z>0\quad \forall N \geq 1$
    \item $N\cdot z''(N)+2\cdot z'(N) = N \cdot c_z \cdot \alpha_z \cdot (\alpha_z - 1) \cdot N^{\alpha_z - 2} + 2 \cdot c_z \cdot \alpha_z \cdot N^{\alpha_z - 1} = c_z \cdot \alpha_z \cdot (\alpha_z+1) \cdot N^{\alpha_z - 1} > 0\quad \forall \alpha_z,c_z>0\quad \forall N \geq 1$
    \item Setting $z(N_1)=1$ leads to $c_z \cdot N_1^{\alpha_z}-c_z = 1 \Leftrightarrow N_1^{\alpha_z} =\frac{c_z + 1}{c_z} \Leftrightarrow N_1 = \left(\frac{c_z + 1}{c_z}\right)^{\frac{1}{\alpha_z}} > 1 \qquad \forall c_z,\alpha_z>0$
\end{enumerate} 
 \end{proof}

\color{black}	

\end{document}